\documentclass[a4paper,twoside,reqno,11pt,tbtags]{article}

\usepackage{amssymb,amsmath,amsthm,mathrsfs}
\usepackage{graphicx}
\usepackage{times}
\usepackage{color}

\leftmargini=\baselineskip

\newtheoremstyle{localthm}
	{5pt} 
	{5pt} 
	{\sl} 
	{} 
	{\bf} 
	{{\rm.}} 
	{.7em} 
	{} 

\theoremstyle{localthm}
\newtheorem{Theorem}{Theorem}
\newtheorem{Corollary}[Theorem]{Corollary}
\newtheorem{Proposition}[Theorem]{Proposition}
\newtheorem{Lemma}[Theorem]{Lemma}

\newtheoremstyle{localrem}
	{5pt} 
	{5pt} 
	{\rm} 
	{} 
	{\bf} 
	{{\rm.}} 
	{.7em} 
	{} 

\theoremstyle{localrem}

\def\Ex{\mathop{\mathrm{I\!E}}\nolimits}
\def\Pr{\mathop{\mathrm{I\!P}}\nolimits}

\newcommand{\R}{\mathbb{R}}

\newcommand{\NN}{\mathcal{N}}

\def\Bin{\mathrm{Bin}}

\newcommand{\argmax}{\mathop{\mathrm{arg\,max}}}

\def\bs{\boldsymbol}

\def\hat{\widehat}
\def\tilde{\widetilde}

\def\Fhat{\hat{F}_n}
\def\FhatS{\hat{F}_n^{\rm S}}
\def\FhatL{\hat{F}_n^{\rm L}}
\def\FhatM{\hat{F}_n^{\rm M}}
\def\FhatZ{\hat{F}_n^{\rm Z}}

\def\Bhat{\hat{B}_n}
\def\BhatS{\hat{B}_n^{\rm S}}
\def\BhatL{\hat{B}_n^{\rm L}}
\def\BhatM{\hat{B}_n^{\rm M}}
\def\BhatZ{\hat{B}_n^{\rm Z}}

\def\eps{\varepsilon}

\def\V{\mathbb{V}}

\def\Nn{\bs{N}_{\!n}}
\def\Xn{\bs{X}_{\!n}}
\def\Rn{\bs{R}_{n}}

\begin{document}

\title{Inference on a Distribution Function\\
	from Ranked Set Samples}
\author{Lutz D\"umbgen and Ehsan Zamanzade\\
	(University of Bern and University of Isfahan)}
\date{October 2013, revised July 2018}
\maketitle

\begin{abstract}
Consider independent observations $(X_1,R_1)$, $(X_2,R_2)$, \ldots, $(X_n,R_n)$ with random or fixed ranks $R_i \in \{1,2,\ldots,k\}$, while conditional on $R_i = r$, the random variable $X_i$ has the same distribution as the $r$-th order statistic within a random sample of size $k$ from an unknown continuous distribution function $F$. Such observation schemes are utilized in situations in which ranking observations is much easier than obtaining their precise values. Two well-known special cases are ranked set sampling (\nocite{McIntyre_1952}{McIntyre 1952}) and judgement post-stratification (\nocite{MacEachern_etal_2004}{MacEachern et al.\ 2004}).

Within a general setting including unbalanced ranked set sampling we derive and compare the asymptotic distributions of three different estimators of the distribution function $F$ as $n \to \infty$ with fixed $k$: The stratified estimator of \nocite{Stokes_Sager_1988}{Stokes and Sager (1988)}, the nonparametric maximum-likelihood estimator of \nocite{Kvam_Samaniego_1994}{Kvam and Samaniego (1994)} and a moment-based estimator of \nocite{Chen_2001}{Chen (2001)}. Our functional central limit theorems generalize and refine previous asymptotic analyses. In addition we discuss briefly pointwise and simultaneous confidence intervals for the distribution function $F$ with guaranteed coverage probability for finite sample sizes.

The methods are illustrated with a real data example, and the potential impact of imperfect rankings is investigated in a small simulation experiment. All in all, the moment-based estimator seems to offer a good compromise between efficiency and robustness versus imperfect ranking, in addition to computational efficiency.
\end{abstract}

\vfill


\addtolength{\baselineskip}{0.5\baselineskip}
\newpage

\section{Introduction}
\label{sec:introduction}

Ranked set sampling and judgement post-stratification are both sampling strategies in situations in which ranking several observations is possible and relatively easy without referring to exact values, whereas obtaining complete observations is much more involved. For instance, this occurs often in agriculture or forestry when the quantities of interest are yields on different plots or of different trees. Good overviews of theory and applications of ranked set sampling are given by \nocite{Wolfe_2004, Wolfe_2012}{Wolfe (2004, 2012)} and \nocite{Chen_etal_2004}{Chen et al.\ (2004)}. Let us explain the two sampling schemes just mentioned in a simple hypothetical example: Suppose we want to estimate the distribution of body heights among all men of age 20-25 in a certain population. Whenever we have obtained a precise measurement $X_i$ of such a man, we could compare him to $k-1$ additional young men and note the rank $R_i \in \{1,2,\ldots,k\}$ of $X_i$ within this small group without measuring the heights of the additional men precisely. This sampling scheme is called judgement post-stratification (JPS), see \nocite{MacEachern_etal_2004}{MacEachern et al.\ (2004)}. Alternatively, for each observation we could recruit a group of $k$ young men, rank them with respect to their heigths and then obtain the precise body height $X_i$ of the person with rank $R_i \in \{1,\ldots,k\}$ only. Here the ranks $R_1, R_2, \ldots, R_n$ have been specified in advance. This sampling scheme, called ranked set sampling (RSS), was introduced by \nocite{McIntyre_1952}{McIntyre (1952)}. If the empirical distribution of the ranks $R_i$ is (approximately) uniform on $\{1,\ldots,k\}$, one talks about balanced RSS, otherwise unbalanced RSS. For instance, if we are mainly interested in the upper tail of the distribution of body heights, we could favour larger ranks $R_i$.

In general we consider independent random pairs $(X_1,R_1)$, $(X_2,R_2)$, \ldots, $(X_n,R_n)$ with fixed or random ranks $R_i \in \{1,2,\ldots,k\}$. Conditional on $R_i = r$, the random variable $X_i$ has the same distribution as the $r$-th order statistic of a random sample of size $k$ from $F$. That means, $X_i$ has distribution function
\[
	F_r(x) \ := \ \Pr(X_i \le x \,|\, R_i = r)
	\ = \ B_r(F(x)) ,
\]
where $B_r : [0,1] \to [0,1]$ denotes the distribution function of the beta distribution with parameters $r$ and $k+1-r$. Thus for $p \in [0,1]$,
\[
	B_r(p)
	\ = \ \sum_{i=r}^k \binom{k}{i} p^i (1 - p)^{k-i}
	\ = \ \int_0^p \beta_r(u) \, du
\]
with
\[
	\beta_r(u) \ = \ C_r u^{r-1} (1 - u)^{k-r}
	\quad\text{and}\quad
	C_r \ = \ k \binom{k-1}{r-1} \ = \ k \binom{k-1}{k-r} ,
\]
see \nocite{David_Nagaraja_2003}{David and Nagaraja (2003)}. The vector $\Nn = (N_{nr})_{r=1}^k$ of stratum sizes
\[
	N_{nr} \ := \ \sum_{i=1}^n 1_{[R_i = r]}
\]
plays a key role. In RSS the ranks $R_1, R_2, \ldots, R_n$ and thus the whole vector $\Nn$ are fixed. In JPS, the $R_i$ are independent and uniformly distributed on $\{1,\ldots,k\}$, whence $\Nn$ follows a multinomial distribution $\mathrm{Mult}(n; 1/k, \ldots, 1/k)$.

Several estimators of the c.d.f.\ $F$ have been proposed. Of course one could just ignore the rank information and compute the empirical c.d.f.\ $\Fhat$,
\[
	\Fhat(x) \ := \ \frac{1}{n} \sum_{i=1}^n 1_{[X_i \le x]} .
\]
In the JPS setting this estimator is unbiased and $\sqrt{n}$-consistent. However, the stratified estimator
\[
	\FhatS \ := \ \frac{1}{\#\{r : N_{nr} > 0\}} \sum_{r \,:\, N_{nr} > 0} \hat{F}_{nr}
\]
with the empirical c.d.f.
\[
	\hat{F}_{nr}(x) \ := \ \frac{1}{N_{nr}} \sum_{i = 1}^n 1_{[R_i = r, \, X_i \le x]}
\]
within stratum $\{i : R_i = r\}$ is usually more efficient. It has been introduced and analyzed in a balanced RSS setting by \nocite{Stokes_Sager_1988}{Stokes and Sager (1988)}. Refinements and modifications of this estimator $\FhatS$ in the JPS setting have been proposed by \nocite{Frey_Ozturk_2011}{Frey and Ozturk (2011)} and \nocite{Wang_etal_2012}{Wang et al.\ (2012)}. In particular, these authors consider situations with small or moderate sample sizes so that some stratum sizes $N_{nr}$ may be zero or the empirical c.d.f.s $\hat{F}_{nr}$ may fail to satisfy order relations which are known for their theoretical counterparts $F_r$.

A second approach to estimating the c.d.f.\ $F$ which can also handle empty strata was introduced by \nocite{Kvam_Samaniego_1994}{Kvam and Samaniego (1994)}. They propose to estimate $F(x)$ by maximizing the conditional log-likelihood function
\begin{align*}
	L_n(x,p) \
	&:= \ \sum_{i=1}^n \bigl[ 1_{[X_i \le x]} \log B_{R_i}(p) 
			+ 1_{[X_i > x]} \log(1 - B_{R_i}(p)) \bigr] \\
	&= \ \sum_{r=1}^k N_{nr} \bigl[ \hat{F}_{nr}(x) \log B_r(p) 
			+ (1 - \hat{F}_{nr}(x)) \log(1 - B_r(p)) \bigr]
\end{align*}
of the indicator vector $(1_{[X_i \le x]})_{i=1}^n$, given the rank vector $\Rn = (R_i)_{i=1}^n$. The resulting estimator $\FhatL$ is given by
\[
	\FhatL(x) \
	:= \ \argmax_{p \in [0,1]}
		L_n(x,p) .
\]
\nocite{Huang_1997}{Huang (1997)} provides a detailed asymptotic analysis of this estimator $\FhatL$ in the special setting when $n = k\ell$, $N_{nr} = \ell$ for $1 \le r \le k$, and $\ell \to \infty$.

A third approach, introduced by \nocite{Chen_2001}{Chen (2001)}, is to estimate $F$ by a moment equality for the naive empirical c.d.f.\ $\Fhat$. Note that
\[
	\Ex \bigl( n \Fhat(x) \,\big|\, \Rn \bigr)
	\ = \ \sum_{r=1}^k N_{nr} B_r(F(x)) .
\]
Hence one can estimate $F(x)$ by the unique number $\FhatM(x) \in [0,1]$ such that
\begin{equation}
\label{eq:Chen}
	n \Fhat(x) \ = \ \sum_{r=1}^k N_{nr} B_r(\FhatM(x)) .
\end{equation}
In the RSS setting with proportions $N_{nr}/n$ converging to fixed numbers $\pi_r > 0$ as $n \to \infty$, \nocite{Chen_2001}{Chen (2001)} proves asymptotic normality of $\sqrt{n} \bigl( \FhatM(x) - F(x) \bigr)$ for finitely many points $x$ and shows that the supremum norm of $\FhatM - F$ converges to zero in probability. (Note that \nocite{Chen_2001}{Chen (2001)} formulates the moment equality \eqref{eq:Chen} with $n \pi_r$ in place of $N_{nr}$, but this would introduce an unnecessary estimation bias.)

In Section~\ref{sec:properties} we present some elementary properties of the estimators $\FhatS$, $\FhatL$ and $\FhatM$ and comment briefly on the computation of the latter two. In addition we describe two methods to obtain pointwise and simultaneous confidence intervals for $F$, respectively. The former procedure is just an adaptation of a method by \nocite{Terpstra_Miller_2006}{Terpstra and Miller (2006)} and closely related to the estimator $\FhatM$. Inverting the underlying tests yields honest confidence intervals for any given quantile of $F$ as proposed by \nocite{Balakrishnan_Li_2006}{Balakrishnan and Li (2006)} for balanced RSS. The confidence bands are a generalization of the confidence bands described by \nocite{Stokes_Sager_1988}{Stokes and Sager (1988)}. Here it turns out that the estimator $\FhatM$ is particularly convenient to work with.

Section~\ref{sec:asymptotics} provides a detailed analysis of the asymptotic distribution of the estimators $\FhatS$, $\FhatL$ and $\FhatM$ as $n \to \infty$ while $k$ is fixed and $N_{nr}/n \to_p \pi_r > 0$ for $1 \le r \le k$. Our analyses provide linear stochastic expansions and functional Central Limit Theorems for the processes $\sqrt{n}(\Fhat^{\rm Z} - F)$, ${\rm Z} = {\rm S}, {\rm L}, {\rm M}$. These results generalize the findings of \nocite{Stokes_Sager_1988}{Stokes and Sager (1988)} about $\FhatS$, of \nocite{Huang_1997}{Huang (1997)} about $\FhatL$ in balanced RSS and of \nocite{Chen_2001}{Chen (2001)} and \nocite{Ghosh_Tiwari_2008}{Ghosh and Tiwari (2008)} about $\FhatM$. We obtain explicit expressions for the asymptotic covariance functions of $\sqrt{n}(\Fhat^{\rm Z} - F)$ which enable efficiency considerations. The most important findings are that (i) the estimator $\FhatL$ is always superior to the other two, (ii) the estimators $\FhatS$ and $\FhatM$ are asymptotically equivalent in case of $\pi_1 = \cdots = \pi_k = 1/k$, and (iii) in unbalanced settings the estimator $\FhatS$ can be substantially worse than the other two estimators. Moreover, the efficiency gain of $\FhatL$ over $\FhatM$ is bounded and typically rather small. In addition we analyze the estimators' asymptotic behavior in the tails of the distribution $F$ where they turn out to be essentially equivalent.

A detailed analysis of a real data example is presented in Section~\ref{sec:imperfect.rankings}. It involves population sizes of Swiss municipalities and illustrates that sampling from finite populations without replacement may render our confidence regions conservative, even if the rankings are not perfect. The impact of imperfect rankings itself is investigated in a small simulation study based on the model of \nocite{Dell_Clutter_1972}{Dell and Clutter (1972)}.

The main proofs are deferred to an appendix. Further technical details and additional material, including references to computer code in \nocite{R_2013}{R}, are collected in a supplement.

\section{Computation of the estimators and exact inference}
\label{sec:properties}

\paragraph{Computations.}
In what follows let $X_{(1)} < X_{(2)} < \cdots < X_{(n)}$ be the order statistics of $X_1, X_2, \ldots, X_n$, augmented by $X_{(0)} := -\infty$ and $X_{(n+1)} := \infty$. One can easily verify that for $\mathrm{Z} = \mathrm{S}, \mathrm{M}, \mathrm{L}$, the estimator $\FhatZ$ is constant on each interval $[X_{(y)}, X_{(y+1)})$, $0 \le y \le n$, where $\FhatZ \equiv 0 $ on $[X_{(0)}, X_{(1)})$ and $\FhatZ \equiv 1$ on $[X_{(n)}, X_{(n+1)})$.

While the computation of the stratified estimator $\FhatS$ is straightforward, the estimators $\FhatM$ and $\FhatL$ may be computed numerically by running a suitable bisection algorithm $n-1$ times. Concerning $\FhatM$, note that $\sum_{r=1}^k N_{nr} B_r(p)$ is continuous and strictly increasing in $p \in [0,1]$ with boundary values $0$ and $1$. Hence for $ 1 \le y < n$ and $X_{(y)} \le x < X_{(y+1)}$, the estimator $\FhatM(x)$ is the unique solution $p \in (0,1)$ of $\sum_{r=1}^k N_{nr} B_r(p) = y$.

As to $\FhatL$, the next lemma provides some essential properties of the log-likelihood function $L_n(\cdot,\cdot)$. Its proof is given in the supplement.

\begin{Lemma}
\label{lem:compute.FhatL}
For any $x \in \R$, the function $L_n(x,\cdot) : [0,1] \to [-\infty,0]$ is continuous and continuously differentiable on $(0,1)$. Its derivative $L_n'(x,p) := \partial L_n(x,p)/\partial p$ is strictly decreasing in $p \in (0,1)$ and equals
\[
	L_n'(x,p) \ = \ \sum_{r=1}^k N_{nr} w_r(p) \bigl[ \hat{F}_{nr}(x) - B_r(p) \bigr]
\]
with the auxiliary function
\[
	w_r(p) \ = \ \frac{\beta_r}{B_r(1 - B_r)}(p)
	\ = \ \frac{\beta_r(p)}{B_r(p) B_{k+1-r}(1-p)} .
\]
Moreover, in case of $X_{(1)} \le x < X_{(n)}$, the limits of $L_n'(x,\cdot)$ at the boundary of $(0,1)$ are equal to $L_n'(x,0) = \infty$ and $L_n'(x,1) = -\infty$.
\end{Lemma}

According to this lemma, for $y \in \{1,\ldots,n-1\}$ and $X_{(y)} \le x < X_{(y+1)}$, the value of $\FhatL(x)$ is the unique number $p \in (0,1)$ such that
\[
	\sum_{r=1}^k N_{nr} w_r(p) \bigl[ \hat{F}_{nr}(X_{(y)}) - B_r(p) \bigr] \ = \ 0 .
\]
The computation of $\FhatM$ and $\FhatL$ for one single data set is of similar complexity. There is, however, an important difference: The vector $\bigl( \FhatM(X_{(y)}) \bigr)_{y=1}^{n-1}$ depends solely on the vector $\Nn = (N_{nr})_{r=1}^k$ of stratum sizes. Hence if we want to simulate the conditional distribution of $\FhatM$, given $\Rn$, we have to compute the vector $\bigl( \FhatM(X_{(y)}) \bigr)_{y=1}^{n-1}$ only once. By way of contrast, the vector $\bigl( \FhatL(X_{(y)}) \bigr)_{y=1}^{n-1}$ depends on the whole matrix $(N_{nry})_{1 \le r \le k, 1 \le y \le n}$ of frequencies $N_{nry} = N_{nr} \hat{F}_{nr}(X_{(y)}) = \sum_{i=1}^n 1_{[R_i = r, \, X_i \le X_{(y)}]}$. For given $\Nn$ there are
\[
	\frac{n!}{N_{n1}! \, N_{n2}! \, \cdots \, N_{nk}!}
\]
possibilities for that matrix, and this number grows exponentially with $n$, unless $\Nn$ is extremely unbalanced. As a consequence, for each new data set we have to compute $\FhatL$ anew, even if $\Nn$ remains unchanged.

\paragraph{Basic distributional properties.}
From now on we condition on the rank vector $\Rn = (R_i)_{i=1}^n$. Hence the vector $\Nn = (N_{nr})_{r=1}^k$ of stratum sizes is viewed as a fixed vector, and all probabilities, expectations and distributional statements refer to the conditional distribution of $\Xn = (X_i)_{i=1}^n$, given $\Rn$.

All estimators $\Fhat$, $\FhatS$, $\FhatM$ and $\FhatL$ are distribution-free in the following sense: Let $\Bhat$, $\BhatS$, $\BhatM$ and $\BhatL$ be defined analogously with raw observations from the uniform distribution on $[0,1]$. That means, we replace the random variables $X_1, X_2, \ldots, X_n$ with random variables $\tilde{X}_1, \tilde{X}_2, \ldots, \tilde{X}_n \in [0,1]$ which are independent, and $\tilde{X}_i$ has (conditional) distribution function $B_r$ if $R_i = r$. Then
\[
	\bigl( \FhatZ(x) \bigr)_{x \in \R}
	\quad \text{has the same distribution as} \quad
	\bigl( \hat{B}_n^{\rm Z}(F(x)) \bigr)_{x \in \R} ,
\]
where ${\rm Z} = {\rm S}, {\rm M}, {\rm L}$. Consequently it suffices to analyze the distribution of the random processes $\bigl( \hat{B}_n^{\rm Z}(t) \bigr)_{t \in [0,1]}$.

\paragraph{Pointwise confidence intervals.}
Recall that the estimator $\FhatM(x)$ was defined by matching $n \Fhat(x)$ to its (conditional) mean. Comparing $n \Fhat(x)$ with its distribution function yields exact confidence bounds for $F(x)$. This approach has been used by \nocite{Terpstra_Miller_2006}{Terpstra and Miller (2006)} in the framework of balanced ranked set sampling. In the present general framework this method works as follows: The (conditional) distribution of $n \Fhat(x)$ depends only on $\Nn$ and $F(x)$. Precisely, in case of $F(x) = p$, it has the same distribution as $\sum_{r=1}^k Y_{r,p}$ with independent random variables $Y_{1,p}$, $Y_{2,p}$, \ldots, $Y_{k,p}$, where
\[
	Y_{r,p} \ \sim \ \Bin(N_{nr}, B_r(p)) .
\]
Let $G_{\Nn,p}$ be the corresponding distribution function, i.e.
\[
	G_{\Nn,p}(y) \ := \ \Pr \Bigl( \sum_{r=1}^k Y_{r,p} \le y \Bigr) .
\]
This is not a standard distribution but a convolution of binomial distributions which can be computed numerically quite easily. Elementary considerations reveal that for any $y \in \{0,1,\ldots,n-1\}$, the distribution function $G_{\Nn,p}(y)$ is continuous and strictly decreasing in $p \in [0,1]$ with boundary values $G_{\Nn,0}(y) = 1$ and $G_{\Nn,1}(y) = 0$. Further, $G_{\Nn,p}(n) = 1$ and $G_{\Nn,p}(-1) = 0$ for all $p \in [0,1]$. Consequently, non-asymptotic p-values for the null hypotheses ``$F(x) \ge p$'' and ``$F(x) \le p$'' are given by $G_{\Nn,p}(n \Fhat(x))$ and $1 - G_{\Nn,p}(n \Fhat(x) - 1)$, respectively. These imply two different $(1 - \alpha)$-confidence regions for $F(x)$, namely,
\begin{align*}
	\bigl\{ p \in [0,1] : G_{\Nn,p}(n \Fhat(x)) \ge \alpha \bigr\} \
	&= \ \bigl[ 0, b_\alpha(\Nn,n \Fhat(x)) \bigr] , \\
	\bigl\{ p \in [0,1] : G_{\Nn,p}(n \Fhat(x) - 1) \le 1 - \alpha \bigr\} \
	&= \ \bigl[ a_\alpha(\Nn,n \Fhat(x)), 1 \bigr] .
\end{align*}
Here $b_\alpha(\Nn,y)$ is the unique solution $p \in (0,1)$ of the equation $G_{\Nn,p}(y) = \alpha$ if $0 \le y \le n-1$, and $b_\alpha(\Nn,n) = 1$. Likewise, $a_\alpha(\Nn,y)$ is the unique solution $p \in (0,1)$ of the equation $G_{\Nn,p}(y-1) = 1 - \alpha$ if $1 \le y \le n$, and $a_\alpha(\Nn,0) = 0$. Obviously one can combine lower and upper bounds and compute the \nocite{Clopper_Pearson_1934}{Clopper and Pearson (1934)} type $(1 - \alpha)$-con\-fid\-ence interval $\bigl[ a_{\alpha/2}(\Nn,n \Fhat(x)), b_{\alpha/2}(\Nn,n\Fhat(x)) \bigr]$ for $F(x)$.

Note that the computation of all these confidence bounds for $F$ boils down to determining only finitely many values $a_\lambda(\Nn,y)$ and $b_\lambda(\Nn,y)$ for $\lambda = \alpha, \alpha/2$ and $y \in \{0,1,\ldots,n\}$.

If we would ignore the ranks $R_i$ and just pretend that $X_1, X_2, \ldots, X_n$ are i.i.d.\ with distribution function $F$, then we would work with the distribution function $G_{n,p}$ of the binomial distribution $\Bin(n,p)$ instead of $G_{\Nn,p}$. This would lead to the traditional confidence bounds $a_{\alpha}^{\rm st}(n,n\Fhat(x))$, $b_{\alpha}^{\rm st}(n,n\Fhat(x))$ and the confidence interval of \nocite{Clopper_Pearson_1934}{Clopper and Pearson (1934)} with endpoints $a_{\alpha/2}^{\rm st}(n,n\Fhat(x))$, $b_{\alpha/2}^{\rm st}(n,n\Fhat(x))$ for $F(x)$.

\paragraph{Confidence bands.}
We may compute Kolmogorov-Smirnov type confidence bands for the unknown distribution function $F$ as follows: Let $\kappa_{}^{\rm Z}(\Nn,\alpha)$ be the $(1 - \alpha)$-quantile of the random variable $\|\hat{B}_n^{\rm Z} - B\|_\infty = \sup_{t \in [0,1]} \bigl| \hat{B}_n^{\rm Z}(t) - t \bigr|$. Then we may conclude with confidence $1 - \alpha$ that
\[
	F(x) \ \in \ \bigl[ \FhatZ(x) \pm \kappa_{}^{\rm Z}(\Nn,\alpha) \bigr]
	\quad\text{for all} \ x \in \R .
\]
The quantiles $\kappa_{}^{\rm Z}(\Nn,\alpha)$ may be estimated via Monte Carlo simulations. As explained before, this procedure is particularly convenient to implement for the moment-matching estimator $\FhatM$, whereas for the likelihood estimator $\FhatL$ it would be very computer-intensive.

\paragraph{Numerical example.}
Figure~\ref{fig:ConfBand} shows for $n = 210$ and $\Nn = (70,70,70), (100,70,40)$ the estimator value $\FhatM(X_{(y)})$ and the twosided $95\%$-confidence bounds $a_{2.5\%}(\Nn,y)$, $b_{2.5\%}(\Nn,y)$, $a_{2.5\%}^{\rm st}(n,y)$ and $b_{2.5\%}^{\rm st}(n,y)$ as a function of $y \in \{0,1,\ldots,n\}$. One sees that the additional rank information leads to more accurate confidence bounds in the balanced setting. In the unbalanced situation, ignoring the rank information and pretending the $X_i$ to be i.i.d.\ would induce a severe bias, and the coverage probabilities would be substantially smaller than $95\%$.

\begin{figure}[h]
\includegraphics[width=0.49\textwidth]{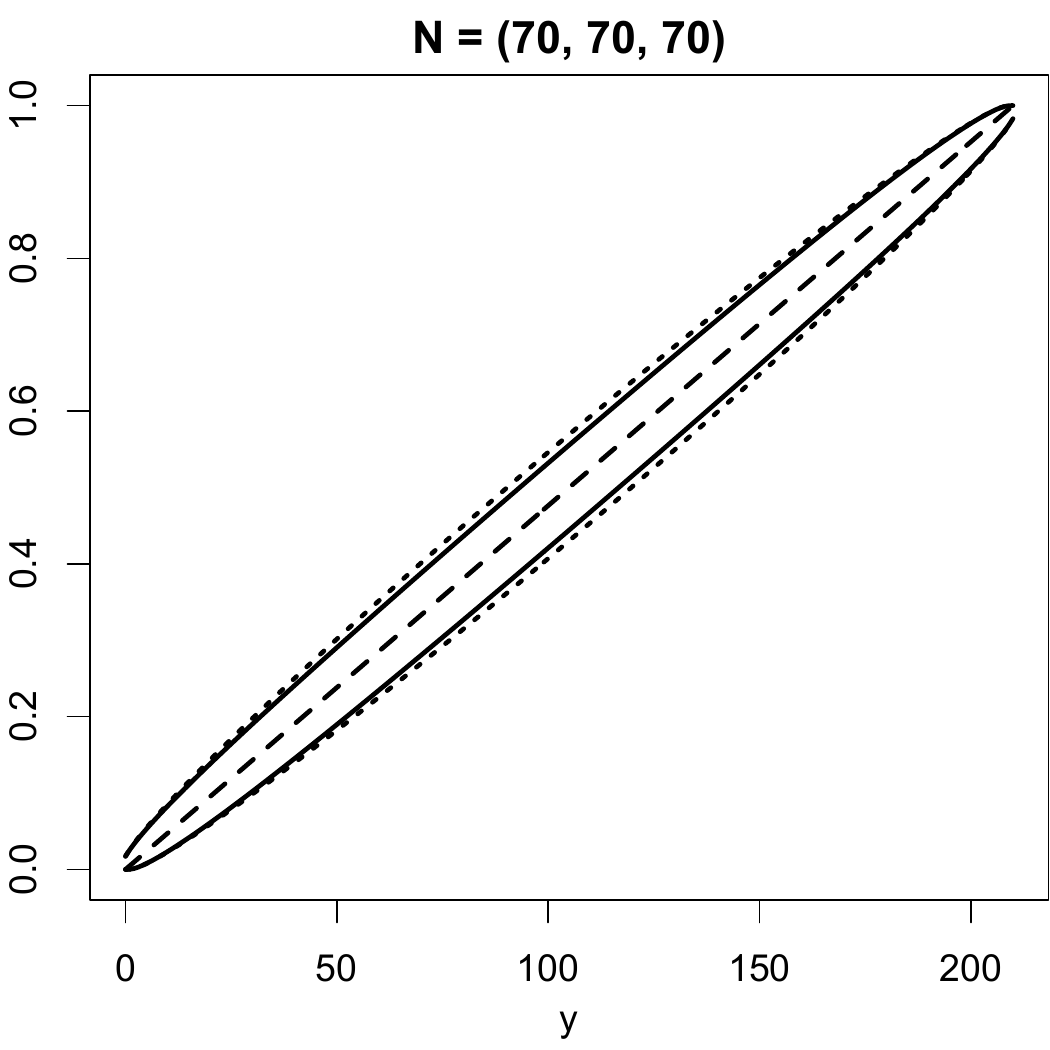}
\hfill
\includegraphics[width=0.49\textwidth]{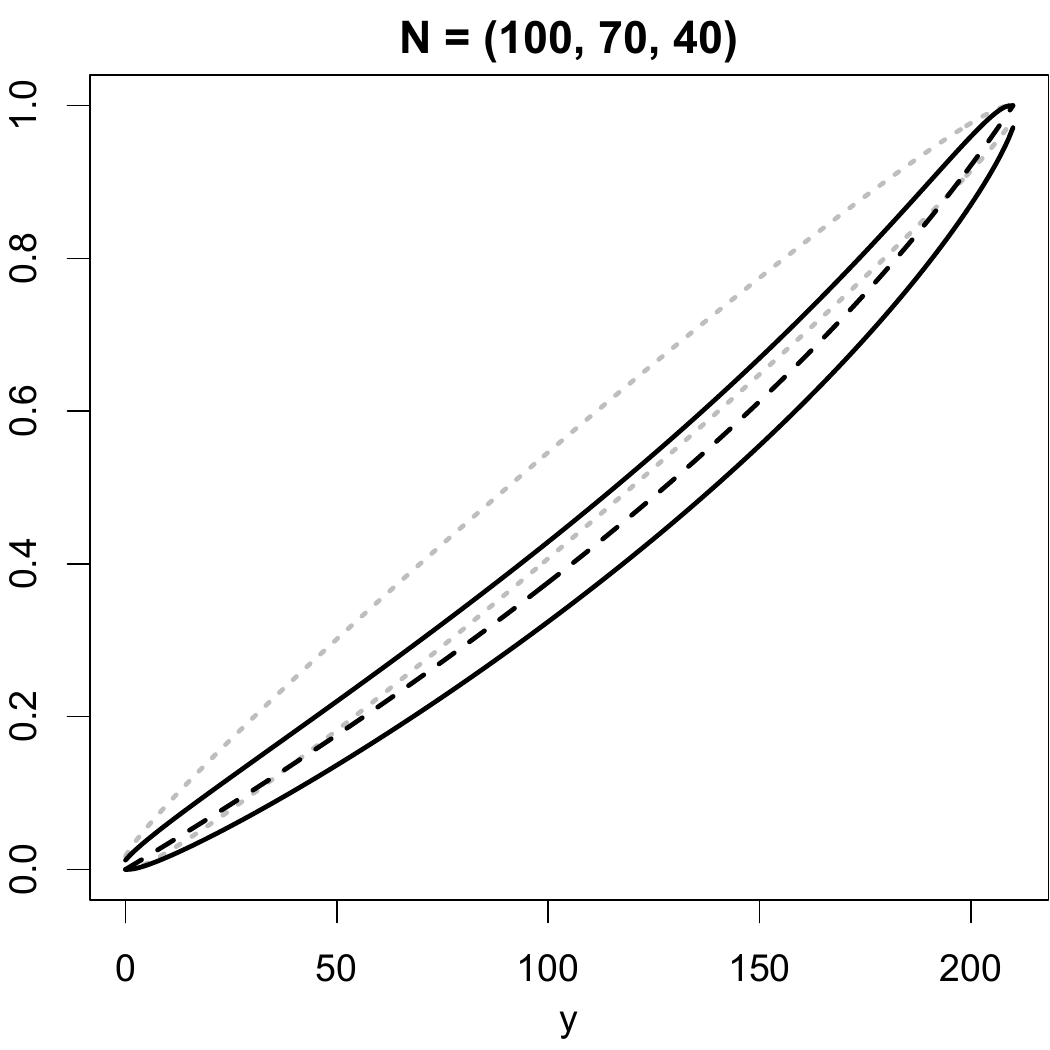}

\caption{Estimator $\FhatM$ and pointwise $95\%$-confidence intervals for $F$: For $y \in \{0,1,\ldots,n\}$ one sees the value $\FhatM(X_{(y)})$ (dashed), the exact confidence bounds $a_{2.5\%}(\Nn,y)$ and $b_{2.5\%}(\Nn,y)$ (solid), and the classical bounds $a_{2.5\%}^{\rm st}(n,y)$ and $b_{2.5\%}^{\rm st}(n,y)$ (dotted).}
\label{fig:ConfBand}
\end{figure}

For Kolmogorov-Smirnov type confidence bands centered at $\FhatM$ we estimated the quantiles $\kappa_{}^{\rm M}(\Nn,5\%)$ in $10^5$ Monte Carlo simulations and obtained
\[
	\hat{\kappa}_{}^{\rm M}(\Nn,5\%) \ = \ \begin{cases}
		0.0790 & \text{for} \ \Nn = (70,70,70), \\
		0.0812 & \text{for} \ \Nn = (100,70,40).
	\end{cases}
\]
For the usual Kolmogorov-Smirnov confidence band with $n = 210$ observations, the critical value would be $\kappa(n,5\%) = 0.0927$.

\paragraph{Unequal group sizes.}
The point estimators $\FhatL, \FhatM$ and the confidence regions just described may be extended easily to a more general setting with independent observations $(X_i,R_i,k_i)$, $1 \le i \le n$, where $k_i \ge 1$ is a fixed integer, $R_i$ is a fixed or random rank in $\{1,2,\ldots,k_i\}$, and
\[
	\Pr(X_i \le x \,|\, R_i = r)
	\ = \ B_{r,k_i+1-r}(F(x)) ,
\]
see for instance \nocite{Bhoj_2001}{Bhoj (2001)} or \nocite{Chen_2001}{Chen (2001)}. Here $B_{r,s}$ denotes the distribution function of the beta distribution with parameters $r$ and $s$.

\section{Asymptotic considerations}
\label{sec:asymptotics}

We consider the asymptotic behavior of the estimators $\BhatS$, $\BhatM$ and $\BhatL$ for fixed $k$ as $n \to \infty$ and
\[
	\pi_{nr} := \frac{N_{nr}}{n} \ \to \ \pi_r
	\quad\text{for} \ 1 \le r \le k .
\]
Recall that we condition on the rank vector $\Rn$. The former condition is satisfied with $\pi_r = 1/k$ both in \nocite{Huang_1997}{Huang's (1997)} setting and in the JPS setting almost surely. In general we assume that
\[
	\begin{cases}
		\pi_1, \ldots, \pi_k > 0
		& \text{in connection with} \ \BhatS , \\
		\pi_1, \pi_k > 0
		& \text{in connection with} \ \BhatM, \BhatL .
	\end{cases}
\]

\paragraph{Linear expansions and limit theorems.}
In what follows let
\[
	\V_{nr} \ := \ \sqrt{N_{nr}} (\hat{B}_{nr} - B_r) \circ B_r^{-1}
\]
for $1 \le r \le k$. Each stochastic process $\V_{nr}$ has the same distribution as a standardized empirical distribution function of $N_{nr}$ independent random variables with uniform distribution on $[0,1]$, see also the appendix. Moreover, the processes $\V_{n1}, \ldots, \V_{nk}$ are stochastically independent. Our first result shows that the three estimators $\BhatS$, $\BhatM$ and $\BhatL$ may be approximated by simpler processes involving $\V_{n1}, \ldots, \V_{nk}$.

\begin{Theorem}[Linear expansion]
\label{thm:Asymptotics1}
For ${\rm Z} = {\rm S}, {\rm M}, {\rm L}$ and any fixed $\delta \in [0,1/2)$,
\[
	\sup_{t \in (0,1)}
		\frac{ \bigl| \sqrt{n} (\hat{B}_n^{\rm Z}(t) - t) - \V_n^{\rm Z}(t) \bigr|}
		     {t^\delta(1-t)^\delta}
	\ \to_p \ 0 ,
\]
where
\[
	\V_n^{\rm Z}(t) \ := \ \sum_{r=1}^k \gamma_{nr}^{\rm Z}(t) \, \V_{nr}(B_r(t))
\]
with continuous functions $\gamma_{n1}^{\rm Z},\ldots,\gamma_{nk}^{\rm Z} : [0,1] \to [0,\infty)$. Precisely, for $t \in (0,1)$,
\begin{align*}
	\gamma_{nr}^{\rm S}(t) \
	&:= \ \frac{1}{k \sqrt{\pi_{nr}}} , \\
	\gamma_{nr}^{\rm M}(t) \
	&:= \ \sqrt{\pi_{nr}} \bigr/ \sum_{s=1}^k \pi_{ns} \beta_s(t) , \\
	\gamma_{nr}^{\rm L}(t) \
	&:= \sqrt{\pi_{nr}}\, w_r(t)
		\Big/ \sum_{s=1}^k \pi_{ns} w_s(t) \beta_s(t)
\end{align*}
with $w_r = \beta_r/(B_r(1 - B_r))$. Moreover,
\[
	\sup_{t \in (0,c] \cup [1-c,1)} \,
		\frac{|\V_n^{\rm Z}(t)|}{t^\delta (1 - t)^\delta}
	\ \to_p \ 0
	\quad\text{as} \ n \to \infty \ \text{and}Ê\ c \downarrow 0 .
\]
\end{Theorem}

The next theorem shows that all estimators $\FhatS, \FhatM, \FhatL$ are asymptotically equivalent in the tail regions. Moreover, the asymptotic behavior in the left and right tail is driven mainly by the processes $\V_{n1}$ and $\V_{nk}$, respectively.

\begin{Theorem}[Linear expansion in the tails]
\label{thm:Asymptotics1B}
For ${\rm Z} = {\rm S}, {\rm M}, {\rm L}$ and any fixed $\kappa \in [1/2,1)$,
\begin{align*}
	\sup_{t \in (0,c]}
		\frac{\bigl| \sqrt{n} (\hat{B}_n^{\rm Z}(t) - t) - \V_{n}^{(\ell)}(t) \bigr|}
		     {t^\kappa}
	\ &\to_p \ 0
\intertext{and}
	\sup_{t \in [1-c,1)}
		\frac{\bigl| \sqrt{n} (\hat{B}_n^{\rm Z}(t) - t) - \V_{n}^{(r)}(t) \bigr|}
		     {(1 - t)^\kappa}
	\ &\to_p \ 0
\end{align*}
as $n \to \infty$ and $c \downarrow 0$, where
\[
	\V_n^{(\ell)}(t) \ := \ \frac{\V_{n1}(B_1(t))}{k \sqrt{\pi_{n1}}}
	\quad\text{and}\quad
	\V_n^{(r)}(t) \ := \ \frac{\V_{nk}(B_k(t))}{k \sqrt{N_{nk}/n}}
\]
\end{Theorem}

It follows from Donsker's theorem for the empirical process that $\V_{nr}$ behaves asymptotically like a standard Brownian bridge process $\V = (\V(u))_{u \in [0,1]}$. Together with Theorem~\ref{thm:Asymptotics1} this leads to the following limit theorem:

\begin{Corollary}[Asymptotic distribution]
\label{cor:Asymptotics}
For ${\rm Z} = {\rm S}, {\rm M}, {\rm L}$, the stochastic process $\V_n^{\rm Z}$ converges in distribution in the space $\ell_\infty([0,1])$ to a centered Gaussian process $\V^{\rm Z}$ with continuous paths on $[0,1]$. Precisely, for $t \in [0,1]$,
\[
	\V^{\rm Z}(t) \ = \ \sum_{r=1}^k \gamma_r^{\rm Z}(t) \V_r(B_r(t))
\]
with independent standard Brownian bridges $\V_1, \ldots, \V_k$ and continuous functions $\gamma_1^{\rm Z}, \ldots, \gamma_k^{\rm Z} : [0,1] \to [0,\infty)$ given by
\begin{align*}
	\gamma_r^{\rm S}(t) \
	&:= \ \frac{1}{k \sqrt{\pi_r}} , \\
	\gamma_r^{\rm M}(t) \
	&:= \ \sqrt{\pi_r} \bigr/ \sum_{s=1}^k \pi_s \beta_s(t) , \\
	\gamma_r^{\rm L}(t) \
	&:= \ \begin{cases}
		\displaystyle
		\sqrt{\pi_r}\, w_r(t)
			\Big/ \sum_{s=1}^k \pi_s w_s(t) \beta_s(t)
			& \text{for} \ 0 < t < 1 , \\
		\sqrt{\pi_r} \, r / (\pi_1 k)
			& \text{for} \ t = 0 , \\[1ex]
		\sqrt{\pi_r} (k+1-r) / (\pi_k k)
			& \text{for} \ t = 1 .
		\end{cases}
\end{align*}
\end{Corollary}

Theorem~\ref{thm:Asymptotics1} and Corollary~\ref{cor:Asymptotics} show that all three estimators $\FhatS, \FhatM, \FhatL$ are root-$n$-consistent. In the asymptotically balanced case with
\begin{equation}
\label{eq:balanced}
	\pi_1 = \pi_2 = \cdots = \pi_k = 1/k ,
\end{equation}
one can easily deduce from $\sum_{s=1}^k \beta_s \equiv k$ that
\[
	\gamma_r^{\rm M} \ \equiv \ \gamma_r^{\rm S} \ = \ 1 / \sqrt{k}
	\quad\text{for} \ 1 \le r \le k .
\]
Hence in this particular case the estimators $\FhatS$ and $\FhatM$ are asymptotically equivalent. But otherwise $\FhatS$ may be substantially worse than $\FhatM$, as shown later.

\paragraph{Relative asymptotic efficiencies.}
Let $K$ be the covariance function of a standard Brownian bridge $\V$, i.e.\ $K(s,t) = \min\{s,t\} - st$ for $s,t \in [0,1]$. Then the covariance function $K^{\rm Z}$ of the Gaussian process $\V^{\rm Z}$ in Corollary~\ref{cor:Asymptotics} is given by
\[
	K^{\rm Z}(s,t) \ = \ \sum_{r=1}^k \gamma_r^{\rm Z}(s) \gamma_r^{\rm Z}(t)
		K \bigl( B_r(s), B_r(t) \bigr) .
\]
In particular, for $0 < t < 1$ the asymptotic distribution of $\sqrt{n} \bigl( \Bhat^{\rm Z}(t) - t \bigr)$ equals $\mathcal{N} \bigl( 0, K^{\rm Z}(t) \bigr)$ with $K^{\rm Z}(t) := K^{\rm Z}(t,t)$ given by
\begin{align*}
	K^{\rm S}(t) \
	&= \ \sum_{r=1}^k \frac{B_r(t)(1 - B_r(t))}{k^2 \pi_r} , \\
	K^{\rm M}(t) \
	&= \ \sum_{r=1}^k \pi_r B_r(t)(1 - B_r(t))
		\Big/ \Bigl( \sum_{s=1}^k \pi_s \beta_s(t) \Bigr)^2 , \\
	K^{\rm L}(t) \
	&= \ \sum_{r=1}^k \pi_r w_r(t)^2 B_r(t)(1 - B_r(t))
		\Big/ \Bigl( \sum_{s=1}^k \pi_s \beta_s(t) w_s(t) \Bigr)^2 \\
	&= \ 1 \Big/ \sum_{s=1}^k \pi_s \beta_s(t) w_s(t) .
\end{align*}
The latter equation follows from $w_r = \beta_r/(B_r(1 - B_r))$. The next result provides a detailed comparison of these asymptotic variances.

\begin{Theorem}[Relative asymptotic efficiencies]
\label{thm:Asymptotics3}
For arbitrary $t \in (0,1)$,
\[
	K^{\rm L}(t) \ \le \ K^{\rm S}(t)
\]
with equality for at most one $t \in (0,1)$. Furthermore,
\[
	K^{\rm L}(t) \ \le \ K^{\rm M}(t)
\]
with equality if, and only if, $t = 1/2$ and $k = 2$. On the other hand,
\begin{align*}
	\sup_{\pi} \frac{K^{\rm S}(t)}{K^{\rm L}(t)} \
	&= \ \infty , \\
	\sup_{\pi} \frac{K^{\rm M}(t)}{K^{\rm L}(t)} \
	&= \ \frac{\rho(t) + \rho(t)^{-1} + 2}{4} \ \le \ \frac{k + k^{-1} + 2}{4} ,
\end{align*}
where the suprema are over all tuples $(\pi_r)_{r=1}^k$ with strictly positive components summing to one, and
\[
	\rho(t) \ := \ \max_{r=1,\ldots,k} w_r(t) \Big/ \min_{r=1,\ldots,k} w_r(t)
	\ \le \ k .
\]
\end{Theorem}

\paragraph{Numerical examples.}
In case of $k = 2$, the upper bound for $K^{\rm M}(t)/K^{\rm L}(t)$ equals $9/8 = 1.125$. More precisely,
\[
	\frac{\rho(t) + \rho(t)^{-1} + 2}{4}
	\ = \ 1 + \frac{u^2}{9 - u^2}
	\ \le \ 1.125
\]
with $u := 2t - 1 \in [-1,1]$, see the supplement for more details.

In case of $k = 3$, the upper bound for $K^{\rm M}(t)/K^{\rm L}(t)$ equals $4/3 \approx 1.333$. Figures~\ref{fig:Asymptotics.k3a} and \ref{fig:Asymptotics.k3b} show the asymptotic variance functions $K(\cdot)$ of $\Bhat$ and $K^{\rm Z}(\cdot)$ of $\BhatZ$ for ${\rm Z} = {\rm S}, {\rm M}, {\rm L}$ in the balanced and one unbalanced situation. Note that in the balanced setting, $\BhatS \equiv \BhatM$ and thus $K^{\rm S}(\cdot) \equiv K^{\rm M}(\cdot)$. In addition one sees the asymptotic relative efficiencies
\[
	E^{\rm Z}(t) \ := \ \frac{K^{\rm Z}(t)}{K^{\rm L}(t)}
\]
of $\BhatL$ versus $\BhatZ$ together with the upper bound
\[
	E_{\rm max}^{\rm M}(t) \ := \ \bigl( \rho(t) + \rho(t)^{-1} + 2 \bigr) / 4
\]
for $E^{\rm M}(t)$. One sees clearly that the inefficiency of $\BhatM$ versus $\BhatL$ is moderate whereas the inefficiency of $\BhatS$ may become substantial in unbalanced settings. Note also that in case of $\pi_1 > \pi_2 > \pi_3$ the accuracy in the left tail increases at the expense of larger errors in the right tail.

\begin{figure}
\includegraphics[width=0.49\textwidth]{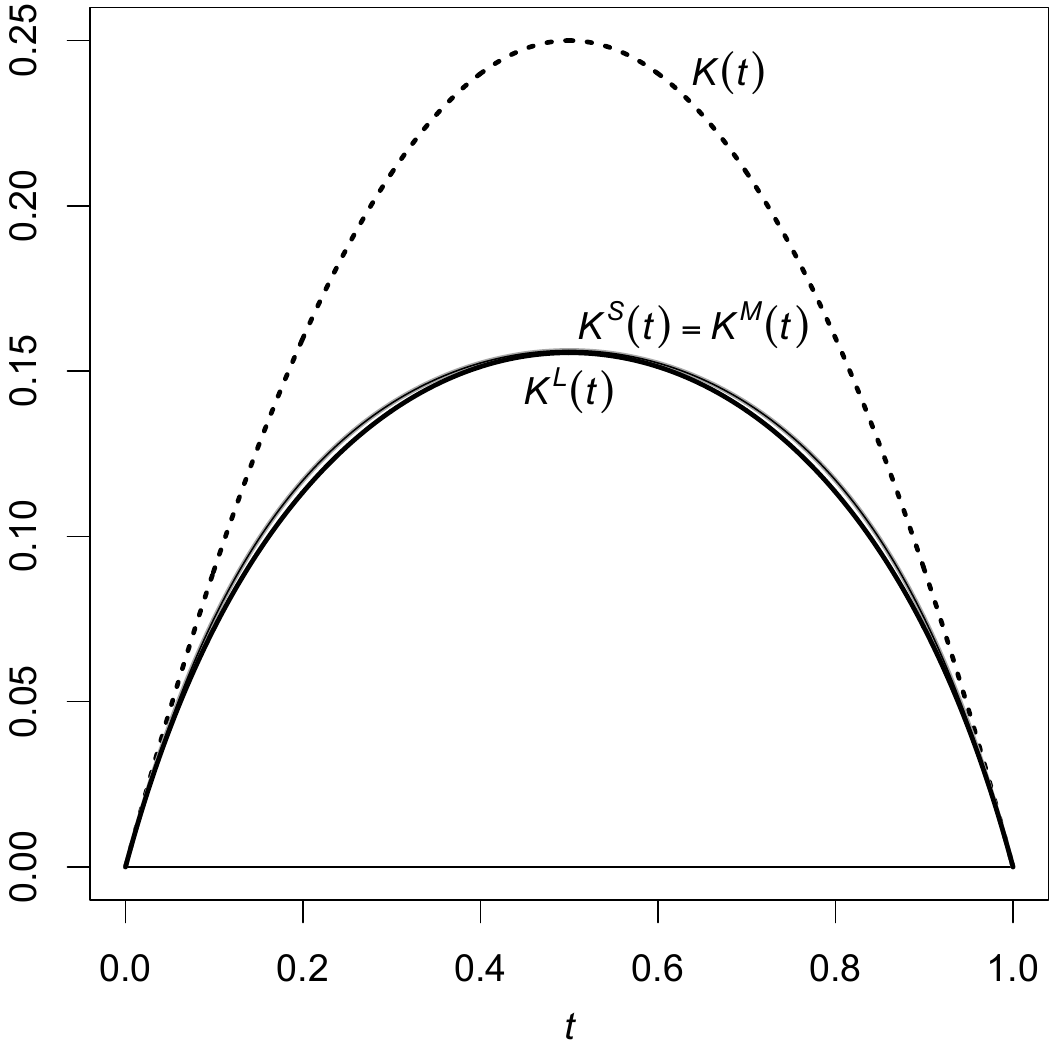}
\hfill
\includegraphics[width=0.49\textwidth]{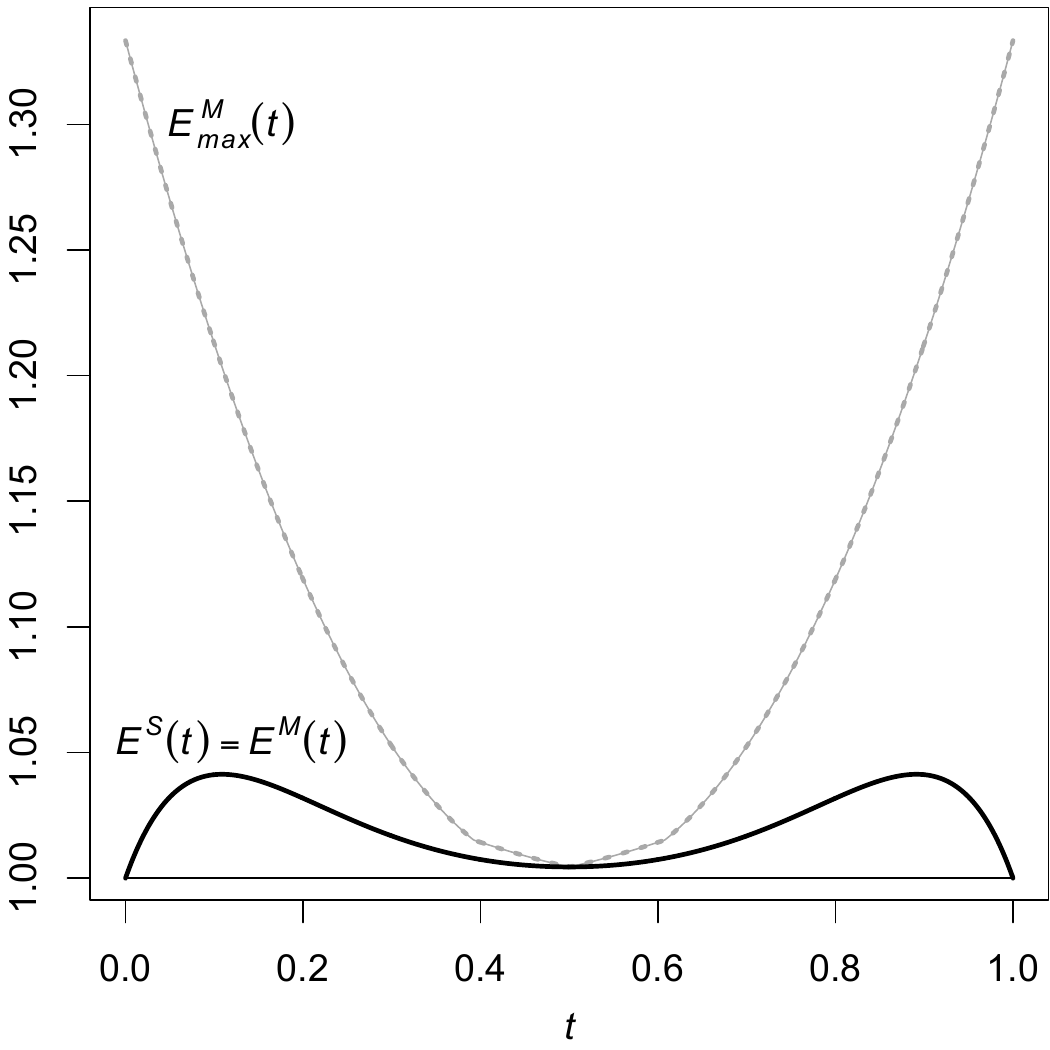}

\caption{Asymptotic variances of $\BhatL$, $\BhatS \equiv \BhatM$, $\Bhat$ (left panel) and relative efficiencies of $\BhatL$ versus $\BhatZ$ (right panel) in case of $\pi_1 = \pi_2 = \pi_3 = 1/3$.}
\label{fig:Asymptotics.k3a}
\end{figure}

\begin{figure}
\includegraphics[width=0.49\textwidth]{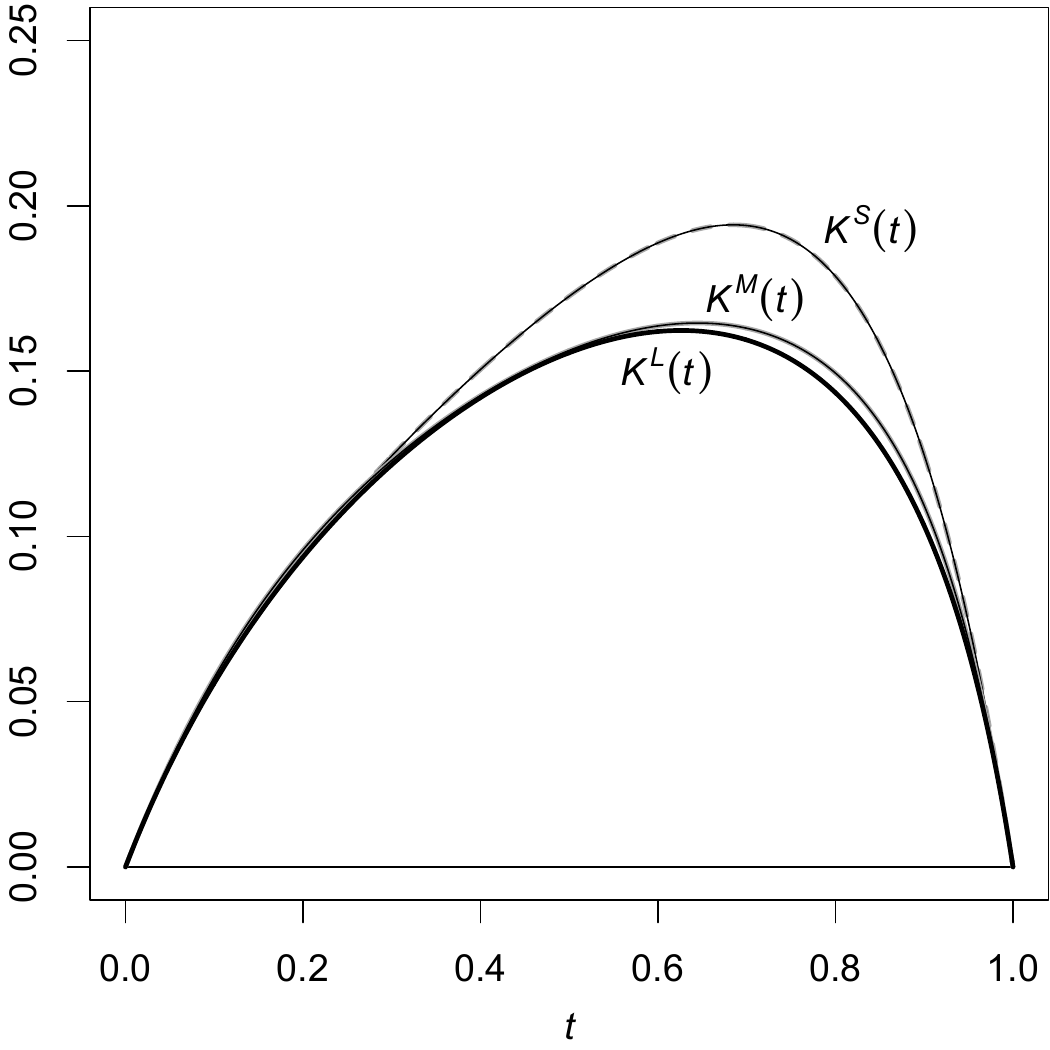}
\hfill
\includegraphics[width=0.49\textwidth]{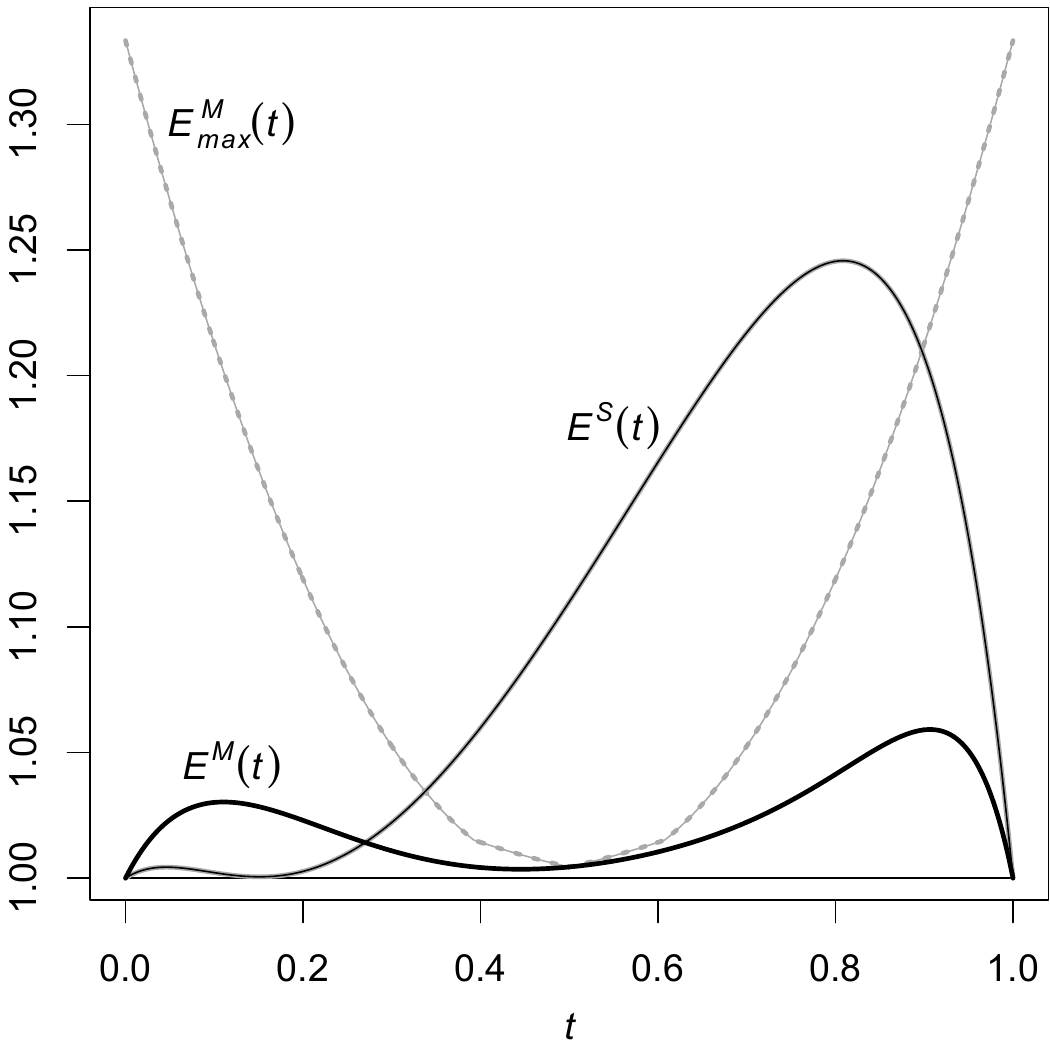}

\caption{Asymptotic variances of $\BhatS$, $\BhatM$, $\BhatL$ (left panel) and relative efficiencies of $\BhatL$ versus $\BhatZ$ (right panel) in case of $(\pi_1,\pi_2,\pi_3) = (10/21,7/21,4/21)$.}
\label{fig:Asymptotics.k3b}
\end{figure}

\paragraph{Implications for confidence intervals.}
One can deduce from Corollary~\ref{cor:Asymptotics} that $n^{1/2} \kappa_{}^{\rm Z}(\Nn,\alpha)$ converges to the $(1 - \alpha)$-quantile of the random supremum norm $\|\mathbb{V}^{\rm Z}\|_\infty$. Moreover, for any $x \in \R$ with $0 < F(x) < 1$, the pointwise confidence bounds satisfy
\begin{align}
\label{eq:Asymptotics.pw1}
	a_\alpha(\Nn,n \Fhat(x)) \
	&= \ \FhatM(x) - \frac{\sqrt{K^{\rm M}(F(x))}}{\sqrt{n}} \, \Phi^{-1}(1 - \alpha)
		+ o_p(n^{-1/2}) \\
\label{eq:Asymptotics.pw2}
	b_\alpha(\Nn,n \Fhat(x)) \
	&= \ \FhatM(x) + \frac{\sqrt{K^{\rm M}(F(x))}}{\sqrt{n}} \, \Phi^{-1}(1 - \alpha)
		+ o_p(n^{-1/2})
\end{align}
with $\Phi^{-1}$ denoting the standard Gaussian quantile function, see the supplement.

\section{A real data example and imperfect rankings}
\label{sec:imperfect.rankings}

\subsection{Population sizes of Swiss municipalities}
\label{subsec:Municip_CH}

Every five years, the Swiss Federal Office of Statistics releases data about all municipalities of Switzerland, including their population sizes. There are currently $2289$ communities, and the two most recent data collections are from 2010 and 2015. Suppose we would have wanted to estimate the distribution function $F$ of population sizes by the end of 2015 in early 2016. Back then only the data of 2010 would have been available, the data of 2015 having been released later in 2016 and corrected in 2017. In principle one could have approached each single municipality to obtain its population size by the end of 2015, but this would have been time-consuming of course. Hence one could have applied RSS sampling as follows: One chooses randomly $n = 210$ disjoint sets of $k = 3$ communities. Within the $i$-th set one determines the unit with rank $R_i$ according to population sizes in 2010 and obtains its precise population size $X_i$ by the end of 2015. The ranks $R_1,\ldots,R_n \in \{1,2,3\}$ are prespecified. If one is particularly interested in smaller municipalities, one could choose $\Rn$ such that, say, $\Nn = (100,70,40)$.

Having the complete data of 2010 and 2015, one can easily simulate this sampling scheme. Figure~\ref{fig:Municip.1} shows for one such sample the estimated distribution function $\FhatM$ together with pointwise and simultaneous $95\%$-confidence intervals as described in Section~\ref{sec:properties}. Since the distribution of population sizes is heavily right-skewed, the horizontal axis shows the decimal logarithms of population sizes. In the lower panel the point estimator $\FhatM$ is replaced with the true distribution function $F$, i.e.\ the empirical distribution function of all $2289$ population sizes in 2015.

\begin{figure}
\centering

\includegraphics[width=0.95\textwidth]{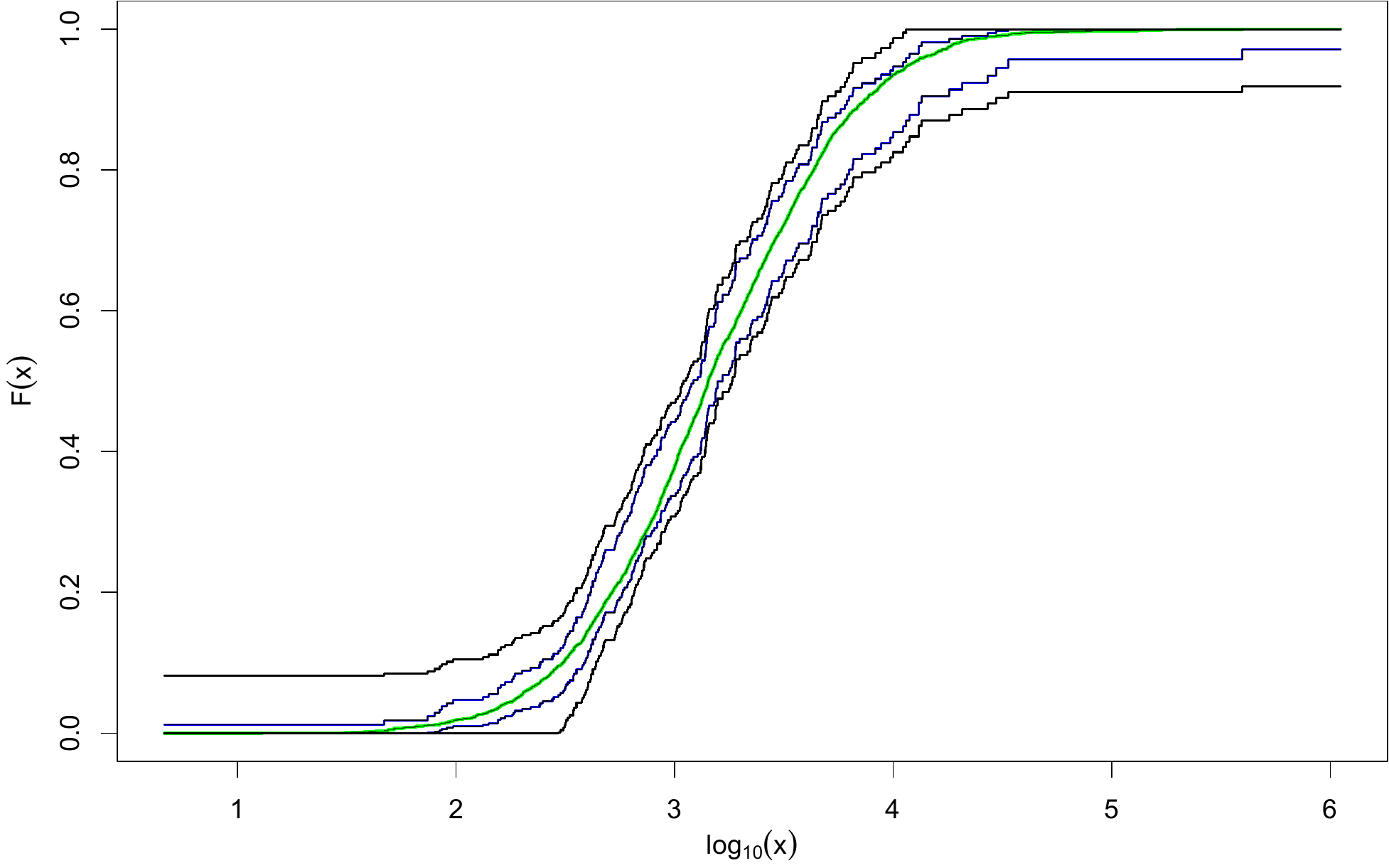}

\caption{Inference about the distribution of population sizes (Section~\ref{subsec:Municip_CH}) with $\Nn = (100,70,40)$. The smoother function is the true c.d.f.\ $F$. The inner and outer two step functions are the pointwise and simultaneous $95\%$-confidence band for $F$.}
\label{fig:Municip.1}
\end{figure}

We simulated this sampling scheme $10^5$ times and analyzed the performance of both $\FhatM$ and the confidence intervals. The Monte Carlo estimator of
\[
	\mathrm{BIAS}(x) \ := \ \Ex \FhatM(x) - F(x)
\]
was everywhere between $- 10^{-4}$ and $10^{-3}$, whereas the MC estimator of
\[
	\mathrm{RMSE}(x) \ := \ \bigl( \Ex \, (\FhatM(x) - F(x))^2 \bigr)^{1/2}
\]
was nowhere larger than $0.0263$. The left panel of Figure~\ref{fig:Municip.2} depicts these two functions $\mathrm{BIAS}$ and $\mathrm{RMSE}$. For each sample and any $x \in \R$ we obtained a pointwise and simultaneous $95\%$-confidence interval, denoted by $C_{\rm pw}(x)$ and $C_{\rm sim}(x)$, respectively. The MC estimator of the error probability $\Pr \bigl( F(x) \not\in C_{\rm pw}(x) \bigr)$ was nowhere larger than $4.22\%$, and the one of $\Pr \bigl( F(x) \not\in C_{\rm sim}(x) \ \text{for some} \ x \in \R \bigr)$ turned out to be smaller than $2.8\%$. The confidence intervals being conservative is probably a consequence of sampling without replacement, which results in more accurate estimators than sampling with replacement. The right panel of Figure~\ref{fig:Municip.2} shows MC estimates of the average widths
\[
	\mathrm{AW}_{\rm pw}(x) \ := \ \Ex \mathrm{width}(C_{\rm pw}(x))
\]
Here one sees clearly the effect of unbalanced sampling with $N_{n1} > N_{n2} > N_{n3}$, the benefit being shorter intervals in the left tail at the expense of longer intervals in the right tail.

\begin{figure}
\includegraphics[width=0.49\textwidth]{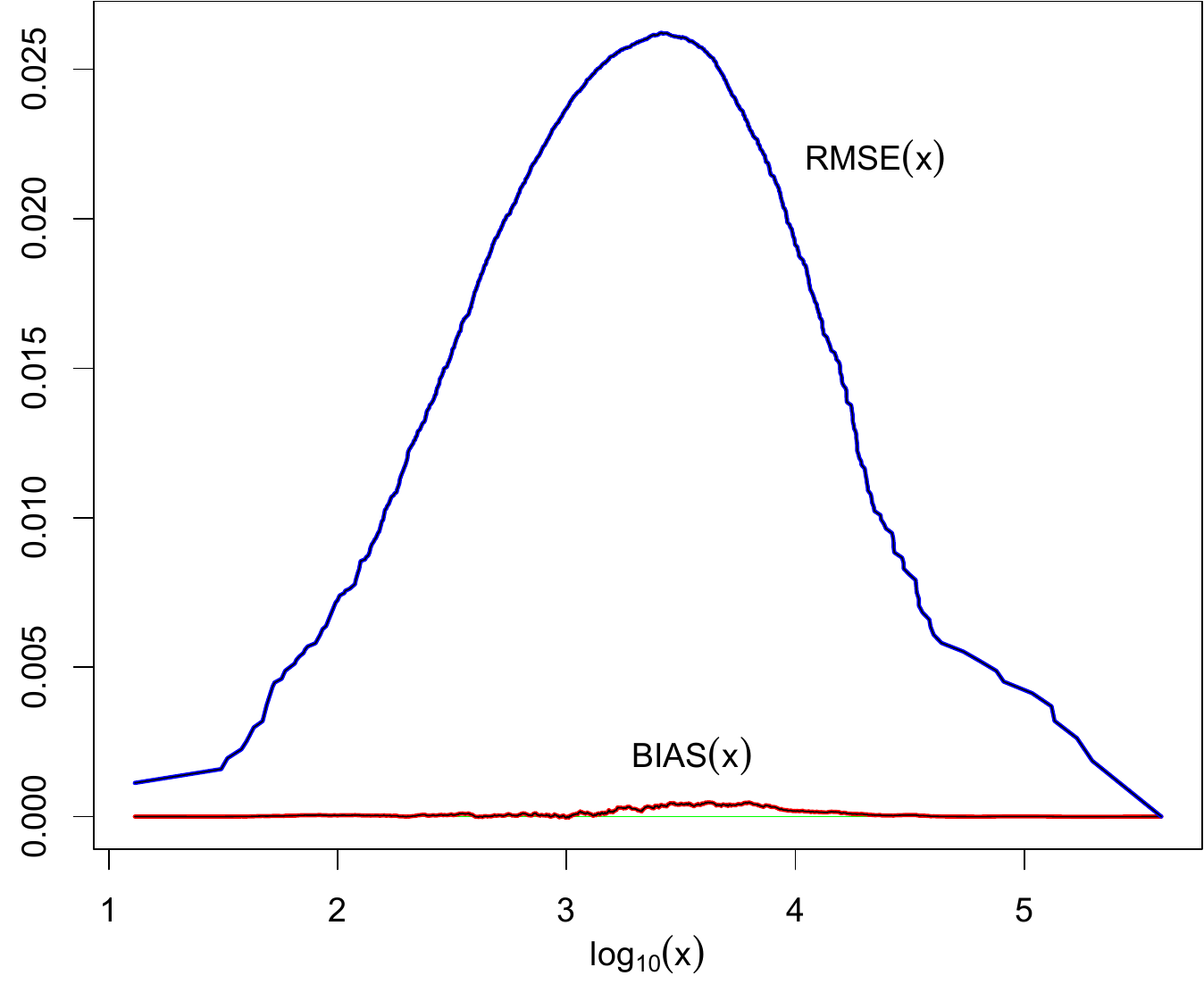}
\hfill
\includegraphics[width=0.49\textwidth]{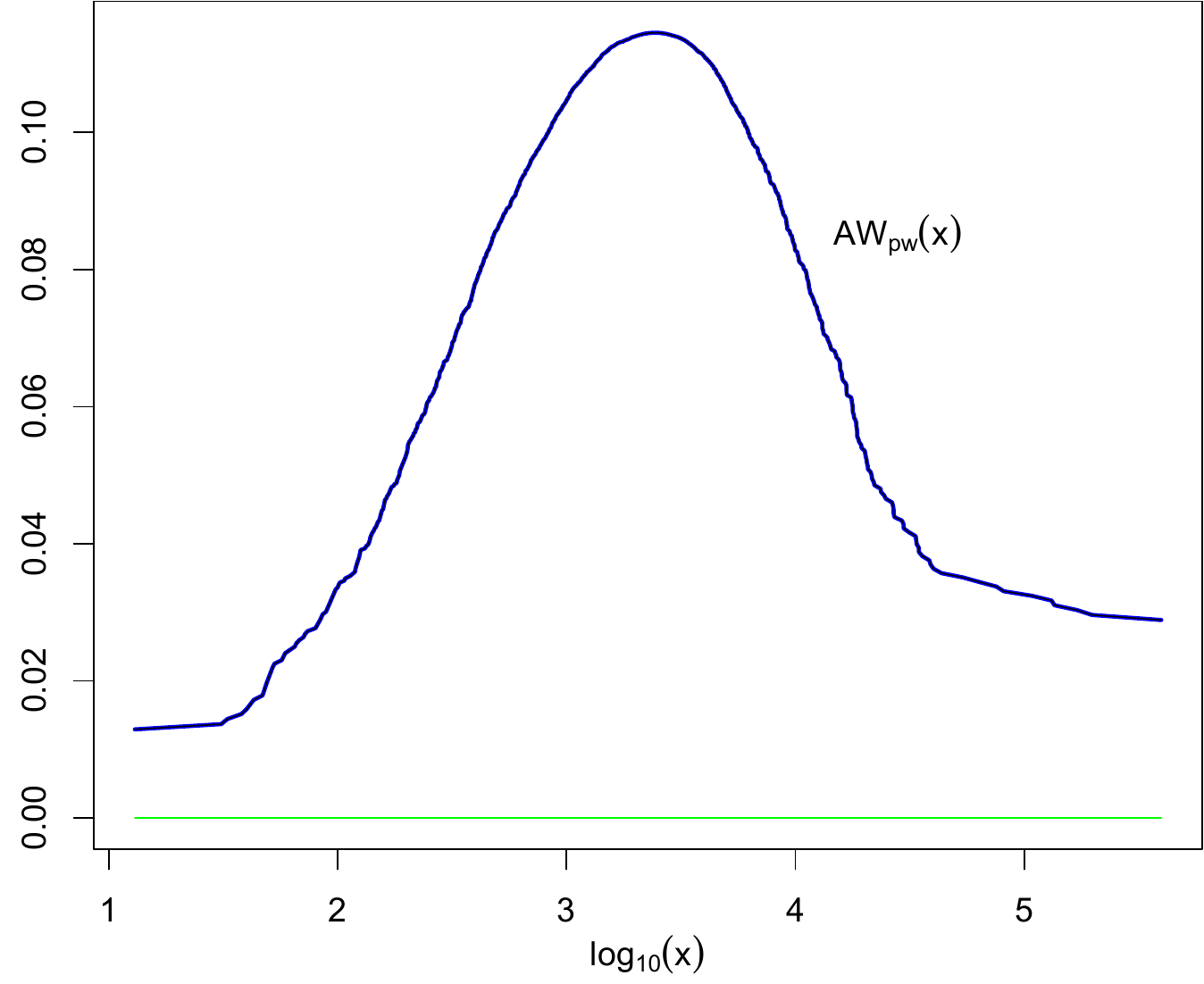}

\caption{Inference about the distribution of population sizes (Section~\ref{subsec:Municip_CH}) with $\Nn = (100,70,40)$. Left panel: bias and root mean squared error of $\FhatM$. Right panel: Average width of pointwise $95\%$-confidence band for $F$.}
\label{fig:Municip.2}
\end{figure}

Note that the ranking of municipalities within the $n$ groups of size $k = 3$ was based on the population sizes in 2010 and thus imperfect. Indeed, a reasonable model for the pairs of log-transformed population sizes in 2010 and 2015 seems to be a bivariate Gaussian distribution with correlation $0.9986$. As a consequence, in our MC simulations the average proportion of imperfect ranks $R_i$ turned out to be $3.1\%$.

Analogous simulations for $k = 4,5$ and different choices of $\Nn$ led to similar results. Enlarging $k$ without changing $n$ leads to larger coverage probabilities, presumably an effect of sampling without replacement, while the modulus of the bias of $\FhatM$ and the proportion of imperfect ranks get larger.

\subsection{Imperfect rankings}
\label{subsec:imperfect.rankings}

In case of sampling \textsl{with} replacement, the previous data example would fit the model of \nocite{Dell_Clutter_1972}{Dell and Clutter (1972)} for ranked set sampling with imperfect rankings quite well. They consider $2nk$ independent random variables $X_{ij} \sim F$ and $\eps_{ij} \sim \NN(0,\tau^2)$ with $1 \le i \le n$ and $1 \le j \le k$. Instead of the true rank of $X_{ij}$ among $X_{i1}, \ldots, X_{ik}$ one obtains the ranks
\[
	R_{ij} \ := \ \sum_{\ell=1}^k 1_{[Y_{i\ell} \le Y_{ij}]}
\]
of the concomitant variables $Y_{ij} := X_{ij} + \eps_{ij}$. If $\sigma > 0$ denotes the standard deviation of the $X_{ij}$, the correlation between $X_{ij}$ and $Y_{ij}$ equals $\rho = (1 + \tau^2/\sigma^2)^{-1/2}$. Finally we obtain for $1 \le i \le n$ the observation $(X_i,R_i) = (X_{i1}, R_{i1})$ in JPS and $(X_{iJ(i)}, R_i)$ in RSS, where $J(i)$ is the unique index in $\{1,\ldots,k\}$ such that $R_{iJ(i)} = R_i$.

In this model the stratified estimator $\FhatS$ is still unbiased, see \nocite{Presnell_Bohn_1999}{Presnell and Bohn (1999) for the RSS setting with $N_{n1},\ldots,N_{nk} > 0$ and \nocite{Dastbaravarde_etal_2016}{Dastbaravarde et al.\ (2016)} for the JPS setting. For that reason we considered $\FhatS$ as a gold standard in our simulation study: We simulated $10^5$ RSS data sets from this model with standard Gaussian distribution function $F = \Phi$, sample size $n = 210$ and different options for $\Nn$ and $\rho$. With these simulations we estimated the bias and root mean squared error,
\[
	\mathrm{BIAS}^{\rm Z}(x)
	\ := \ \Ex \FhatZ(x) - F(x)
	\quad\text{and}\quad
	\mathrm{RMSE}^{\rm Z}(x)
	\ := \ \bigl( \Ex \, (\FhatZ(x) - F(x))^2 \bigr)^{1/2} ,
\]
for ${\rm Z} = {\rm S}, {\rm M}, {\rm L}$. In addition we estimated the relative efficiency
\[
	\mathrm{RE}^{\rm Z}(x) \ := \ \mathrm{RMSE}^{\rm S}(x)^2 / \mathrm{RMSE}^{\rm Z}(x)^2
\]
of $\FhatZ$ versus the stratified estimator $\FhatS$.

Firstly we considered $N_{n1} = N_{n2} = N_{n3} = 70$. Here $\FhatS \equiv \FhatM \equiv \hat{F}_n^{}$. In Figure~\ref{fig:Dell.Clutter.1} one sees on the left hand side the functions $\mathrm{BIAS}^{\rm L}$ and $\mathrm{RMSE}^{\rm L}$ for three different values of the correlation $\rho$. While $\FhatS \equiv \FhatM$ is unbiased, the bias of $\FhatL$ gets worse as $\rho$ decreases. For all three estimators $\FhatZ$ the root mean squared error increases as $\rho$ decreases. The right hand side of Figure~\ref{fig:Dell.Clutter.1} depicts the relative efficiency function $\mathrm{RE}^{\rm L}$. As predicted by asymptotic theory, $\mathrm{RE}^{\rm L} > 1$ in case of $\rho = 1$, but for smaller correlations the relative efficiency drops below $1$ in the tails.

\begin{figure}
\includegraphics[width=0.49\textwidth]{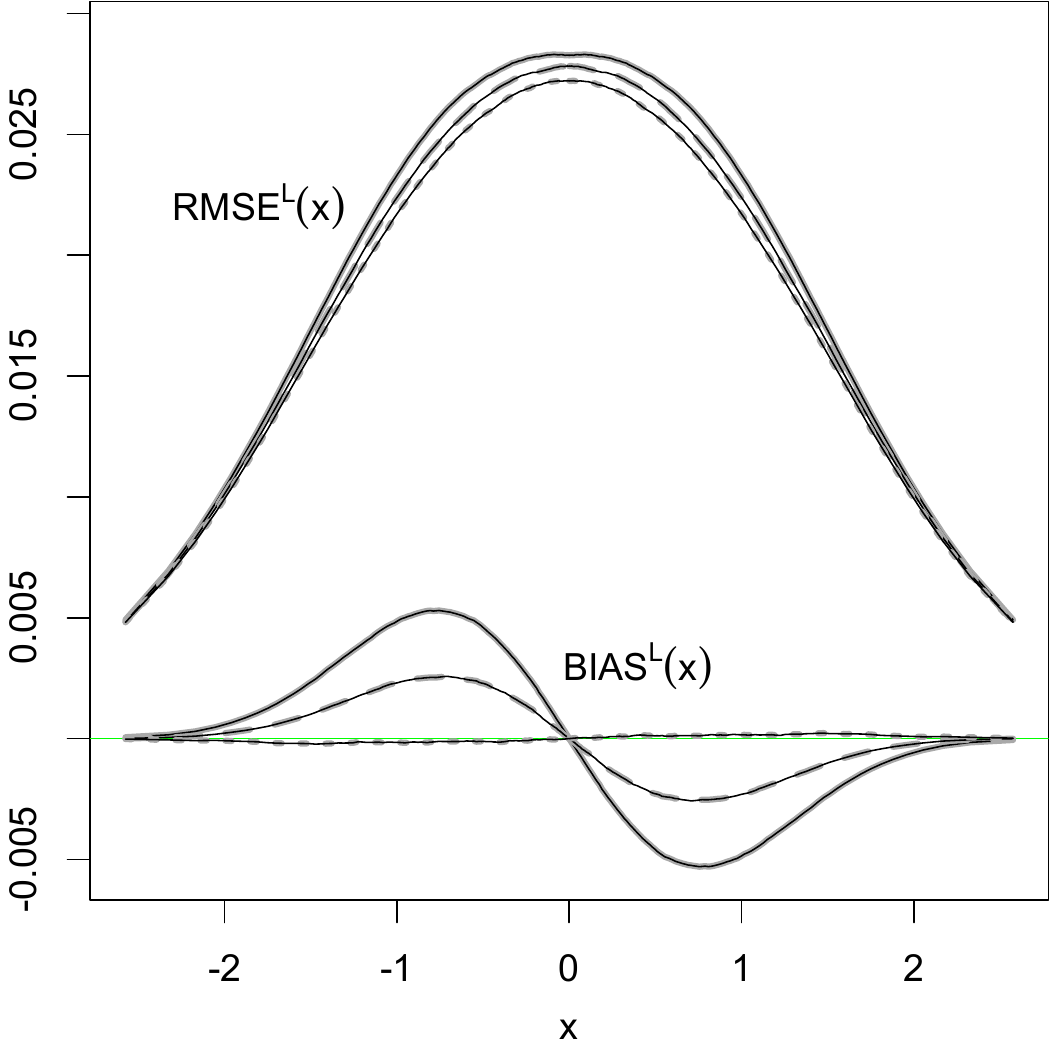}
\hfill
\includegraphics[width=0.49\textwidth]{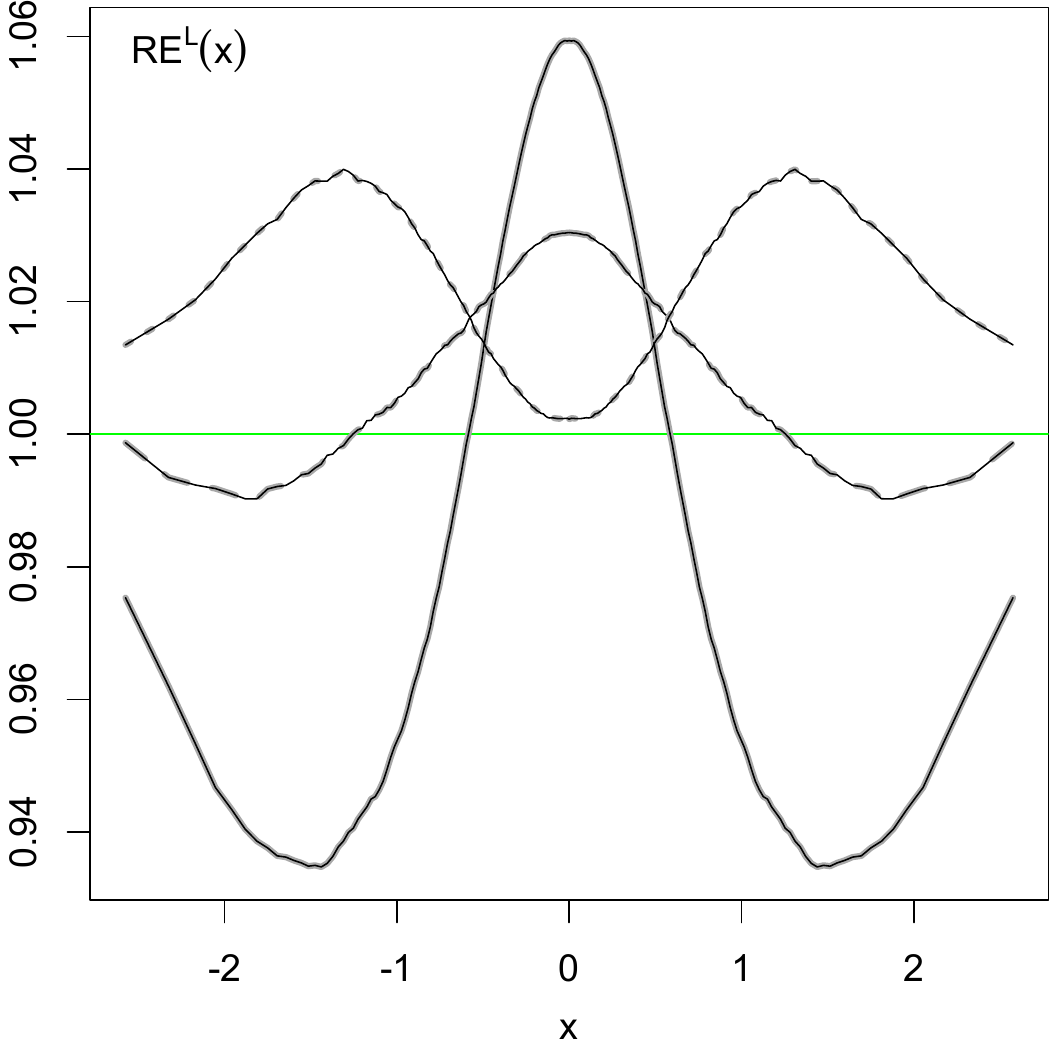}
\caption{Performance of $\FhatL$ in balanced setting with $N_{n1} = N_{n2} = N_{n3} = 70$: Bias and root mean squared error (left panel) and relative efficiency versus $\FhatS$ (right panel) for correlations $\rho = 1$ (dotted), $\rho = 0.95$ (dashed) and $\rho = 0.9$ (solid).}
\label{fig:Dell.Clutter.1}
\end{figure}

Secondly we considered the unbalanced situation with $\Nn = (100, 70, 40)$. Now the three estimators $\FhatZ$ are different, and only $\FhatS$ is unbiased. In Figure~\ref{fig:Dell.Clutter.2} we show bias and root mean squared errors of $\FhatM$ and $\FhatL$. Clearly the bias of of $\FhatM$ and $\FhatL$ gets worse as $\rho$ decreases, where $\FhatM$ is a bit more robust than $\FhatL$. Nevertheless the plots of the relative efficiencies $\mathrm{RE}^{\rm M}$ and $\mathrm{RE}^{\rm L}$ show that for $\rho = 0.95$ the moment-matching estimator outperforms the stratified one everywhere, and also the likelihood estimator is better at most places. For $\rho = 0.9$, the likelihood estimator is less favorable than the other two.

\begin{figure}
\includegraphics[width=0.49\textwidth]{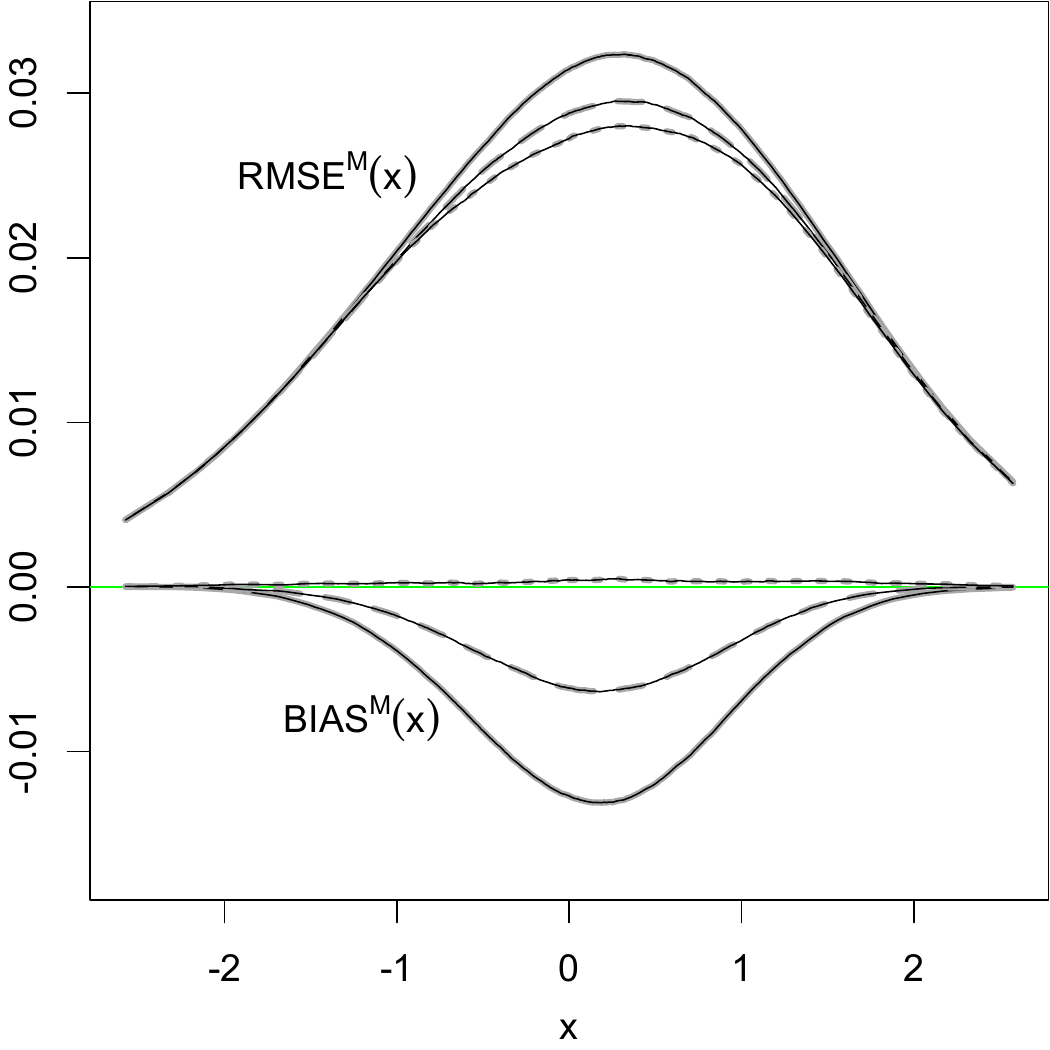}
\hfill
\includegraphics[width=0.49\textwidth]{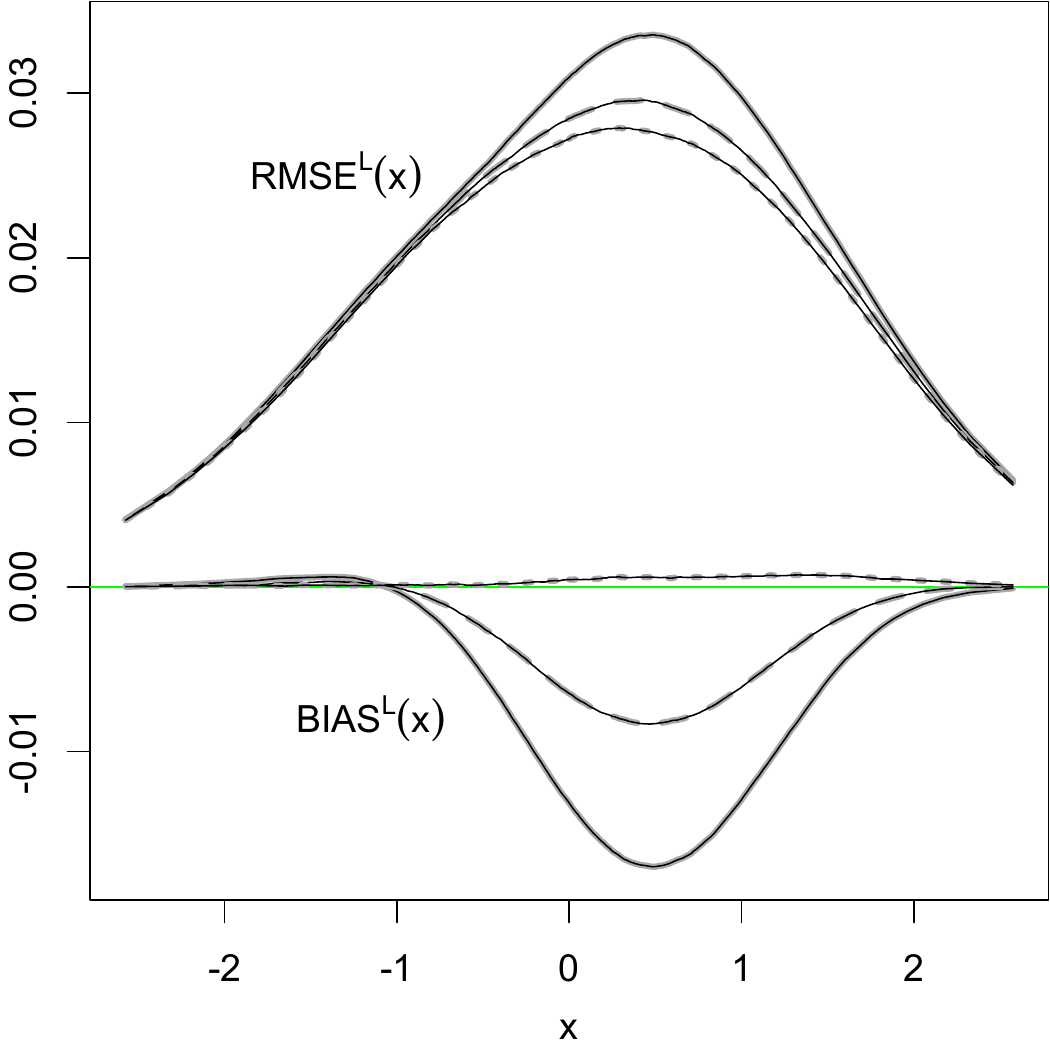}

\includegraphics[width=0.49\textwidth]{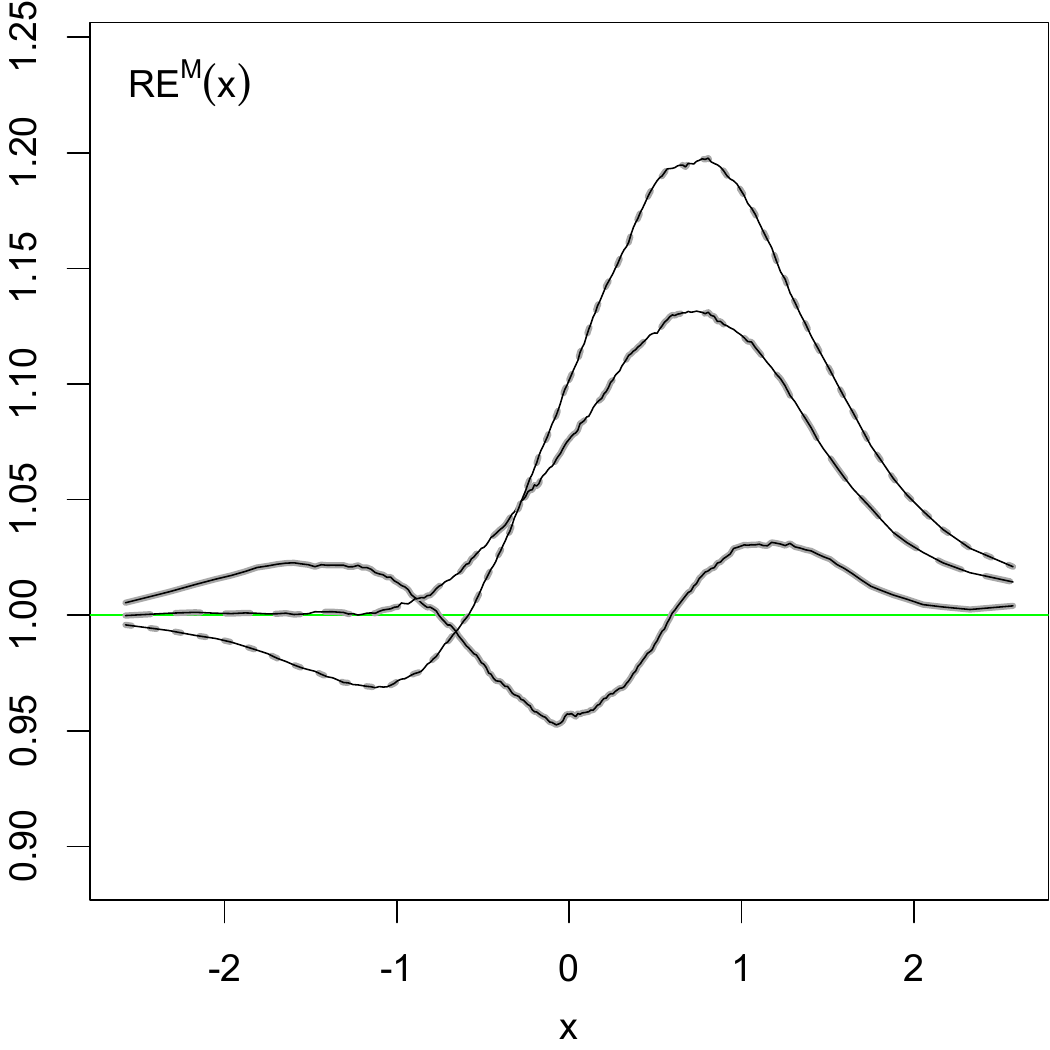}
\hfill
\includegraphics[width=0.49\textwidth]{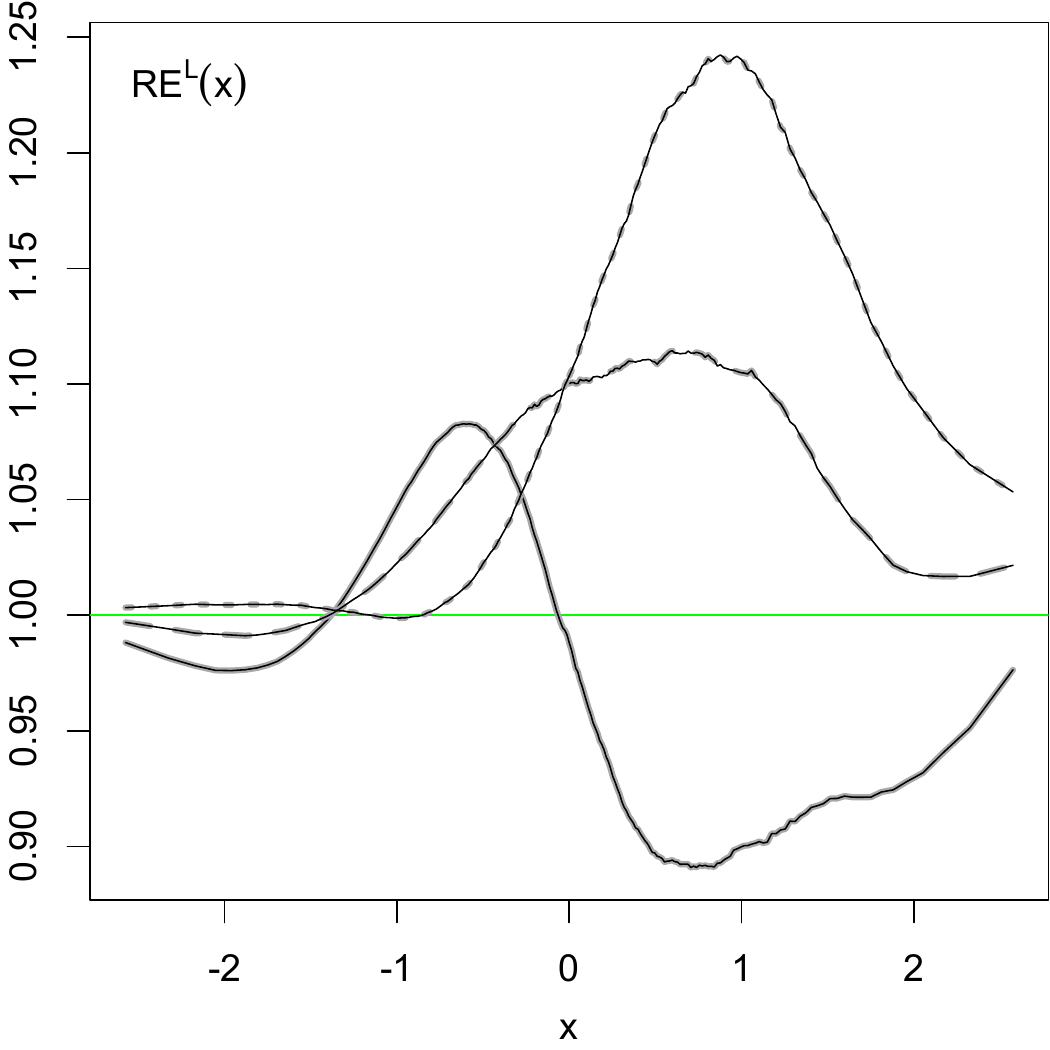}

\caption{Performance of $\FhatM$ and $\FhatL$ in unbalanced setting with $\Nn = (100,70,40)$: Biases and root mean squared errors (upper panels) and relative efficiencies versus $\FhatS$ (lower panels) for correlations $\rho = 1$ (dotted), $\rho = 0.95$ (dashed) and $\rho = 0.9$ (solid).}
\label{fig:Dell.Clutter.2}
\end{figure}

\section*{Conclusions and future research}

The present paper confirms and generalizes previous findings that the estimator $\FhatL$ is the most efficient one in case of perfect ranking, both in balanced and unbalanced situations. In terms of computational efficiency, however, the estimator $\FhatM$ has clear advantages and is particularly convenient as an ingredient for simultaneous confidence bands. Further it is closely related to pointwise confidence bands for $F$. For now we restricted ourselves to Kolmogorov--Smirnov type bands, but other variants might be worthwhile to study.

The simulations in Section~\ref{subsec:imperfect.rankings} indicate that even in case of imperfect rankings, both $\FhatM$ and $\FhatL$ perform well compared to $\FhatS$, as long as the ranking precision is high. While $\FhatL$ appears to be most sensitive to imperfect rankings, $\FhatM$ seems to offer a good compromise in terms of efficiency (for perfect rankings) and robustness against ranking errors. Investigating and understanding these differences thoroughly would be an interesting topic for future research.

\section*{Appendix}

We first recall two well-known facts about uniform empirical processes, see \nocite{Shorack_Wellner_1986}{Shorack and Wellner (1986)}.

\begin{Proposition}
\label{prop:WeightedUniformEP}
Let $U_1, U_2, U_3, \ldots$ be independent random variables with uniform distribution on $[0,1]$. For $N \in \mathbb{N}$ and $u \in [0,1]$ define
\[
	\V^{(N)}(u) \ := \ N^{-1/2} \sum_{i=1}^N \bigl( 1 \{U_i \le u\} - u) .
\]
Then, as $N \to \infty$, $\V ^{(N)}$ converges in distribution in $\ell_\infty([0,1])$ to a standard Brownian bridge $\V$ on $[0,1]$. Moreover, for any fixed $\delta \in [0,1/2)$ and $\epsilon > 0$,
\begin{align*}
	\sup_{N \ge 1} \Pr \biggl(
		\sup_{u \in (0,1)}
			\frac{|\V^{(N)}(u)|}{u^\delta (1 - u)^\delta} \ge C \biggr)
	\ &\to \ 0 \quad\text{as} \ C \uparrow \infty , \\
	\sup_{N \ge 1} \Pr \Bigl(
		\sup_{u \in (0,c] \cup [1-c,1)}
			\frac{|\V^{(N)}(u)|}{u^\delta (1 - u)^\delta} \ge \epsilon \Bigr)
	\ &\to \ 0 \quad\text{as} \ c \downarrow 0 .
\end{align*}
\end{Proposition}

For the estimators $\FhatM$, $\FhatL$ we need some basic facts and inequalities for the auxiliary functions $w_k$ and $B_k$ which are proved in the supplement:

\begin{Lemma}
\label{lem:calculus}
\textbf{(a)} \ For $r = 1,2,\ldots,k$, the function $w_r$ on $(0,1)$ may be written as $w_r(t) = \tilde{w}_r(t) / (t(1-t))$ with $\tilde{w}_r : [0,1] \to (0,\infty)$ continuously differentiable. Moreover, for $r = 1,2,\ldots,k$ and $t \in (0,1)$,
\[
	1 \ \le \ \tilde{w}_r(t) \ \le \ \max(r,k+1-r) .
\]

\noindent
\textbf{(b)} \ For any constant $c \in (0,1)$ there exists a number $c' = c'(k,c) > 0$ with the following property: If $t,p \in (0,1)$ such that
\[
	\frac{|p - t|}{t(1-t)} \ \le \ c ,
\]
then for $r = 1,2,\ldots,k$,
\[
	\max \Bigl\{ \Bigl| \frac{w_r(p)}{w_r(t)} - 1 \Bigr| ,
		\Bigl| \frac{B_r(p) - B_r(t)}{\beta_r(t) (p - t)} - 1 \Bigr| \Bigr\}
	\ \le \ c' \frac{|p - t|}{t(1-t)} .
\]
\end{Lemma}

\begin{proof}[\bf Proof of Theorem~\ref{thm:Asymptotics1}]
We start with the weight functions $\gamma_{nr}^{\rm Z}$: Note that by Lemma~\ref{lem:calculus},
\begin{align*}
	\gamma_{nr}^{\rm S}(t) \
	&= \ \frac{1}{k \sqrt{\pi_{nr}}} , \\
	\gamma_{nr}^{\rm M}(t) \
	&= \ \sqrt{\pi_{nr}} \bigr/ \sum_{s=1}^k \pi_{ns} \beta_s(t) , \\
	\gamma_{nr}^{\rm L}(t) \
	&= \sqrt{\pi_{nr}}\, \tilde{w}_r(t)
		\Big/ \sum_{s=1}^k \pi_{ns} \tilde{w}_s(t) \beta_s(t)
\end{align*}
with the probability weights $\pi_{nr} := N_{nr}/n$ and continuous functions $\tilde{w}_r : [0,1] \to [1,k]$. Since the beta densities $\beta_r$ are also continuous with $\beta_1(0) = \beta_k(1) = k$, this shows that $\gamma_{nr}^{\rm Z}$ is well-defined and continuous, provided that its denominator is strictly positive, i.e.
\[
	\begin{cases}
		\pi_{n1}, \ldots, \pi_{nk} > 0
		& \text{if} \ {\rm Z} = {\rm S} , \\
		\pi_{n1}, \pi_{nk} > 0
		& \text{if} \ {\rm Z} = {\rm M}, {\rm L} .
	\end{cases}
\]
For sufficiently large $n$ this is the case, because $\lim_{n\to\infty} \pi_{nr} = \pi_r$ for all $r$. The functions $\gamma_r^{\rm Z}$ in Corollary~\ref{cor:Asymptotics} are continuous, too, and elementary considerations reveal that
\begin{equation}
\label{eq:Asymptotics gamma}
	\max_{t \in [0,1], \, 1 \le r \le k}
		\bigl| \gamma_{nr}^{\rm Z}(t) - \gamma_r^{\rm Z}(t) \bigr|
	\ \to \ 0
\end{equation}
as $n \to \infty$. In particular, $\max_{t \in [0,1], 1 \le r \le k} \gamma_{nr}^{\rm Z}(t) = O(1)$.

Note that for $n \ge 1$ and $1 \le r \le k$, the empirical process $\V_{nr}$ is distributed as $\V^{(N_{nr})}$ in Proposition~\ref{prop:WeightedUniformEP}. Note also that the distribution functions $B_r$ satisfy $B_1 \ge B_2 \ge \cdots \ge B_k$, because for $1 \le r < k$ the density ratio $\beta_{r+1}/\beta_r$ is a positive multiple of $t/(1 - t)$ and thus strictly increasing. Consequently, for $1 \le r \le k$,
\[
	B_r(t) \ \le \ B_1(t) \ \le \ kt
	\quad\text{and}\quad
	1 - B_r(t) \ \le \ 1 - B_k(t) \ \le \ k(1-t) ,
\]
so
\[
	\frac{B_r(t)(1 - B_r(t))}{t(1-t)} \ \le \ k .
\]
Consequently,
\begin{align*}
	\sup_{t \in (0,1)}
		\frac{|\V_{nr}(B_r(t))|}{t^\delta (1 - t)^\delta} \
	&\le \ k^\delta \sup_{u \in (0,1)}
		\frac{|\V_{nr}(u)|}{u^\delta(1 - u)^\delta}
		\ = \ O_p(1) \quad\text{and} \\
	\sup_{u \in (0,c] \cup [1-c,1)}
		\frac{|\V_{nr}(B_r(t))|}{t^\delta (1 - t)^\delta} \
	&\le \ k^\delta \sup_{u \in (0,kc] \cup [1 - kc,1)}
		\frac{|\V_{nr}(u)|}{u^\delta (1 - u)^\delta}
		\ \to_p \ 0
\end{align*}
as $n \to \infty$ and $c \downarrow 0$. All in all we may conclude that
\begin{align}
\label{eq:AsVn1}
	\sup_{t \in (0,1)} \,
		\frac{|\V_n^{\rm Z}(t)|}{t^\delta (1 - t)^\delta}
	\ &= \ O_p(1) , \\
\label{eq:AsVn2}
	\sup_{t \in (0,c] \cup [1-c,1)} \,
		\frac{|\V_n^{\rm Z}(t)|}{t^\delta (1 - t)^\delta}
	\ &\to_p \ 0
	\quad\text{as} \ n \to \infty \ \text{and}Ê\ c \downarrow 0 .
\end{align}

It remains to be shown that the process $\sqrt{n} (\hat{B}_n^{\rm Z} - B)$ may be approximated by $\V_n^{\rm Z}$. In case of ${\rm Z} = {\rm S}$ it follows from $\sum_{r=1}^k \beta_r \equiv k$ that $\sum_{r=1}^k B_r = k B$, and this implies that
\[
	\sqrt{n} (\BhatS - B)
	\ = \ \sum_{r=1}^k \frac{\sqrt{n} (\hat{B}_{nr} - B_r)}{k}
	\ = \ \sum_{r=1}^k \gamma_{nr}^{\rm S} \, \V_{nr} \circ B_r
	\ = \ \V_n^{\rm S} .
\]

For ${\rm Z} = {\rm M}, {\rm L}$ it suffices to show that for any fixed number $b \ne 0$ and
\[
	p_n^{\rm Z}(t)
	\ := \ t
		+ \frac{\V_n^{\rm Z}(t) + b t^\delta(1-t)^\delta}{\sqrt{n}}
\]
the following statements are true: If $b < 0$, then with asymptotic probability one,
\begin{equation}
\label{eq:b<0}
	\left.\begin{array}{c}
		\displaystyle
		\inf_{t \in (0,1)} \Bigl(
			n \Bhat(t) - \sum_{r=1}^k N_{nr} B_r(p_n^{\rm M}(t)) \Bigr) \\[3ex]
		\displaystyle
		\inf_{t \in (0,1)} \, L_n'(t,p_n^{\rm L}(t))
	\end{array}\right\}
	\ \ge \ 0 .
\end{equation}
If $b > 0$, then with asympototic probability one,
\begin{equation}
\label{eq:b>0}
	\left.\begin{array}{c}
		\displaystyle
		\sup_{t \in (0,1)} \Bigl(
			n \Bhat(t) - \sum_{r=1}^k N_{nr} B_r(p_n^{\rm M}(t)) \Bigr) \\[3ex]
		\displaystyle
		\sup_{t \in (0,1)} \, L_n'(t,p_n^{\rm L}(t))
	\end{array}\right\}
	\ \le \ 0 .
\end{equation}
Here we use the conventions that $L_n'(t,\cdot) := \infty$ and $B_r := 0$ on $(-\infty,0]$ while $L_n'(t,\cdot) := -\infty$ and $B_r := 1$ on $[1,\infty)$.

To verify these claims, we split the interval $(0,1)$ into $(0,c_n]$, $[c_n,1-c_n]$ and $[1-c_n,1)$ with numbers $c_n \in (0,1/2)$ to be specified later, where $c_n \downarrow 0$.

On $[c_n,1-c_n]$ we utilize Lemma~\ref{lem:calculus}: For $t \in [c_n, 1 - t_n]$ and $p \in (0,1)$ such that $|p - t| \le t(1-t)/2$ we may write
\begin{align*}
	n \Bhat(t) \ - & \sum_{r=1}^k N_{nr} B_r(p) \\
	&= \ \sum_{r=1}^k \sqrt{N_{nr}} \V_{nr}(B_r(t))
		- \sum_{r=1}^k N_{nr}(B_r(p) - B_r(t)) \\
	&= \ \sum_{r=1}^k \sqrt{N_{nr}} \V_{nr}(B_r(t))
		- \sum_{r=1}^m N_{nr} \beta_r(t) (p - t)
		+ \rho_n^{\rm M}(t,p) \\
	&= \ \sum_{r=1}^k N_{nr} \beta_r(t)
			\Bigl( \frac{\V_n^{\rm M}(t)}{\sqrt{n}} - (p - t) \Bigr)
		+ \rho_n^{\rm M}(t,p)
\intertext{and}
	L_n'(t,p) \
	&= \ \sum_{r=1}^k \sqrt{N_{nr}} w_r(p) \V_{nr}(B_r(t))
		- \sum_{r=1}^k N_{nr} w_r(p) (B_r(p) - B_r(t)) \\
	&= \ \sum_{r=1}^k \sqrt{N_{nr}} w_r(t) \V_{nr}(B_r(t))
		- \sum_{r=1}^k N_{nr} w_r(t) \beta_r(t) (p - t)
		+ \rho_n^{\rm L}(t,p) \\
	&= \ \sum_{r=1}^k N_{nr} w_r(t) \beta_r(t)
			\Bigl( \frac{\V_n^{\rm L}(t)}{\sqrt{n}} - (p - t) \Bigr)
		+ \rho_n^{\rm L}(t,p) ,
\end{align*}
where
\begin{align*}
	|\rho_n^{\rm M}(t,p)| \
	&\le \ \frac{O(n) |p - t|^2}{t(1-t)} , \\
	|\rho_n^{\rm L}(t,p)| \
	&\le \ \frac{O_p(\sqrt{n}) t^{\delta} (1 - t)^{\delta} |p - t|}{t(1-t)}
		+ \frac{O(n) |p - t|^2}{t^2(1-t)^2} .
\end{align*}
Note that for $t \in [c_n,1-c_n]$,
\[
	\frac{\bigl| p_n^{\rm Z}(t) - t \bigr|}{t(1-t)}
	\ \le \ \frac{O_p(1) t^\delta(1 - t)^{\delta}}{\sqrt{n} \, t(1-t)}
	\ \le \ \frac{O_p(1)}{\sqrt{n} \, c_n^{1-\delta}} .
\]
Hence we choose $c_n$ such that $c_n \downarrow 0$ but $n c_n^{2(1-\delta)} \to \infty$. With this choice we may conclude that uniformly in $t \in [c_n,1-c_n]$,
\begin{align*}
	\bigl| \rho_n^{\rm M}(t,p_n^{\rm M}(t)) \bigr| \
	&\le \ O_p(c_n^{\delta-1}) t^\delta(1-t)^\delta , \\
	\bigl| \rho_n^{\rm L}(t,p_n^{\rm L}(t)) \bigr| \
	&\le \ O_p(c_n^{\delta-1}) t^{\delta-1}(1-t)^{\delta-1} .
\end{align*}
On the other hand, since $\beta_1(t) + \beta_k(t) \ge \beta_1(1/2) + \beta_k(1/2) = k 2^{2-k}$,
\begin{align*}
	\sum_{r=1}^k N_{nr} \beta_r(t) \
	&\ge \ k 2^{2 - k} \min\{N_{n1},N_{nk}\} , \\
	\sum_{r=1}^k N_{nr} w_r(t) \beta_r(t) \
	&\ge \ \frac{k 2^{2-k} c_w}{t(1-t)} \, \min\{N_{n1},N_{nk}\} .
\end{align*}
Consequently,
\begin{align*}
	n \Bhat(t) \ - & \sum_{r=1}^k N_{nr} B_r(p_n^{\rm M}(t)) \\
	&= \ \sum_{r=1}^k N_{nr} \beta_r(t)
		\frac{- b t^\delta(1-t)^\delta}{\sqrt{n}}
		+ \rho_n^{\rm M}(t,p_n^{\rm M}(t)) \\
	&= \ \sum_{r=1}^m N_{nr} \beta_r(t) \frac{t^\delta(1-t)^\delta}{\sqrt{n}}
		\Bigl( - b + O_p(c_n^{\delta-1} n^{-1/2}) \kappa_n^{\rm M}(t) \Bigr) \\
\intertext{and}
	L_n'(t,p_n^{\rm L}(t)) \
	&= \ \sum_{r=1}^k N_{nr} w_r(t) \beta_r(t)
		\frac{-b t^\delta(1-t)^\delta}{\sqrt{n}}
		+ \rho_n^{\rm L}(t,p_n^{\rm L}(t)) \\
	&= \ \sum_{r=1}^k N_{nr} w_r(t) \beta_r(t) \frac{t^\delta(1-t)^\delta}{\sqrt{n}}
		\Bigl( - b + O_p(c_n^{\delta-1} n^{-1/2}) \kappa_n^{\rm L}(t) \Bigr)
\end{align*}
for some random functions $\kappa_n^{\rm M}, \kappa_n^{\rm L} : [c_n,1-c_n] \to [-1,1]$. These considerations show that \eqref{eq:b<0} and \eqref{eq:b>0} are satisfied with $[c_n,1-c_n]$ in place of $(0,1)$.

It remains to verify \eqref{eq:b<0} and \eqref{eq:b>0} with $(0,c_n]$ in place of $(0,1)$; the interval $[1-c_n,1)$ may be treated analogously. Note first that for $2 \le r \le k$,
\[
	B_r(t) \ \le \ B_2(t) \ \le \ k(k-1) t^2/2
	\quad\text{and}\quad
	\beta_r(t) \ \le \ k 2^{k-1} t ,
\]
so
\[
	\bigl| B_r(p) - B_r(t) \bigr|
	\ = \ \Bigl| \int_t^p \beta_r(u) \, du \Bigr|
	\ \le \ O(\max(p,t)) (p - t) .
\]
Futhermore, since $B_1(t) = 1 - (1 - t)^k$,
\[
	B_1(p) - B_1(t) \ = \ k (p - t) + O(\max(t,p)) (p - t) .
\]
Hence for $t \in (0,c_n]$ and $p \in (0,2c_n]$,
\begin{align*}
	n \Bhat(t) \ - & \sum_{r=1}^k N_{nr} B_r(p) \\
	&= \ \sum_{r=1}^k \sqrt{N_{nr}} \V_{nr}(B_r(t))
		- \sum_{r=1}^k N_{nr}(B_r(p) - B_r(t)) \\
	&= \ - N_{n1} k (p - t) + \rho_n^{\rm M}(t,p)
\intertext{and}
	L_n'(t,p) \
	&= \ \sum_{r=1}^k \sqrt{N_{nr}} w_r(p) \V_{nr}(B_r(t))
		- \sum_{r=1}^k N_{nr} w_r(p) (B_r(p) - B_r(t)) \\
	&= \ - N_{n1} w_1(p) k (p - t)
		+ \rho_n^{\rm L}(t,p) ,
\end{align*}
where
\begin{align*}
	|\rho_n^{\rm M}(t,p)| \
	&\le \ o_p(\sqrt{n}) t^\delta + O(n c_n) (p - t) , \\
	|\rho_n^{\rm L}(t,p)| \
	&\le \ o_p(\sqrt{n}) p^{-1} t^\delta + O(n c_n) p^{-1} (p - t) .
\end{align*}

Note also that
\[
	\sup_{t \in (0,c_n]}
		\Bigl| \frac{\sqrt{n}(p_n^{\rm Z}(t) - t)}{t^\delta(1-t)^\delta}
			- b \Bigr|
	\ \to_p \ 0 .
\]
In particular, $\sup_{t \in (0,c_n]} p_n^{\rm Z}(t) = c_n + o_p(n^{-1/2} c_n^\delta) = c_n (1 + o_p(1))$, and in case of $b > 0$, $\Pr \bigl( p_n^{\rm Z}(t) > 0 \ \text{for} \ 0 < t \le c_n \bigr) \to 1$.

In case of $b > 0$, these considerations show that for $0 < t \le c_n$,
\begin{align*}
	n \Bhat(t) \ - & \sum_{r=1}^k N_{nr} B_r(p_n^{\rm M}(t)) \\
	&= \ - N_{n1} k (p_n^{\rm M}(t) - t) + \rho_n^{\rm M}(t,p_n^{\rm M}(t)) \\
	&\le \ \frac{N_{n1} k t^\delta(1-t)^\delta}{\sqrt{n}}
			\bigl( -b + o_p(1) \bigr)
		+ o_p(\sqrt{n}) t^\delta + O(\sqrt{n} c_n) t^\delta \\
	&\le \ \frac{N_{n1} k t^\delta(1-t)^\delta}{\sqrt{n}}
		\bigl( - b + o_p(1) \bigr)
\intertext{and}
	L_n'(t,p_n^{\rm L}(t)) \
	&= \ - N_{n1} w_1(p) k (p_n^{\rm L}(t) - t)
		+ \rho_n^{\rm L}(t,p_n^{\rm Z}(t)) \\
	&\le \ \frac{N_{n1} w_1(p) k t^\delta(1-t)^\delta}{\sqrt{n}}
			\bigl( - b + o_p(1) \bigr)
		+ o_p(\sqrt{n}) p^{-1} t^\delta + O(\sqrt{n} c_n) p^{-1} t^\delta \\
	&\le \ \frac{N_{n1} w_1(p) k t^\delta(1-t)^\delta}{\sqrt{n}}
			\bigl( - b + o_p(1) \bigr) .
\end{align*}
Analogously, in case of $b < 0$, for any $t \in (0,c_n]$ we obtain the inequalities
\begin{align*}
	n \Bhat(t) - \sum_{r=1}^k N_{nr} B_r(p_n^{\rm M}(t)) \
	&\ge \ \begin{cases}
		\displaystyle
		\frac{N_{n1} k t^\delta(1-t)^\delta}{\sqrt{n}}
			\bigl( - b + o_p(1) \bigr)
		& \text{if} \ p_n^{\rm M}(t) > 0 , \\
		0
		& \text{if} \ p_n^{\rm M}(t) \le 0 ,
		\end{cases} \\
	L_n'(t,p_n^{\rm L}(t)) \
	&\ge \ \begin{cases}
		\displaystyle
		\frac{N_{n1} w_1(p) k t^\delta(1-t)^\delta}{\sqrt{n}}
			\bigl( - b + o_p(1) \bigr)
		& \text{if} \ p_n^{\rm L}(t) > 0 , \\
		\infty
		& \text{if} \ p_n^{\rm L}(t) \le 0 .
		\end{cases}
\end{align*}
Hence \eqref{eq:b<0} and \eqref{eq:b>0} are satisfied with $(0,c_n]$ in place of $(0,1)$.
\end{proof}

\begin{proof}[\bf Proof of Theorem~\ref{thm:Asymptotics1B}]
For symmetry reasons it suffices to prove the first part about the left tails. Let $(c_n)_n$ be a sequence of numbers in $(0,1/2]$ converging to zero. Then for $t \in (0,c_n]$ and $\delta := \kappa/2 \in (0,1/2)$,
\[
	\bigl| \sqrt{n} \bigl( \BhatS(t) - t \bigr) - \V_n^{(\ell)}(t) \bigr|
	\ = \ \Bigl| \sum_{r=2}^k \frac{\V_{nr}(B_r(t))}{k \sqrt{N_{nr}/n}} \Bigr|
	\ \le \ t^{2 \delta} o_p(1)
	\ = \ t^\kappa o_p(1) .
\]
Concerning $\BhatM$ and $\BhatL$, for any $t \in (0,c_n]$ and $p \in (0,1)$,
\begin{align*}
	n \Bhat(t) \ - & \sum_{r=1}^k N_{nr} B_r(p) \\
	&= \ \sum_{r=1}^k \sqrt{N_{nr}} \V_{nr}(B_r(t))
		- \sum_{r=1}^k N_{nr}(B_r(p) - B_r(t)) \\
	&= \ \sqrt{N_{n1}} \V_{n1}(B_1(t)) - N_{n1} k(p - t)
		+ \rho_n^{\rm M}(t,p) \\
	&= \ N_{n1} k
			\Bigl( \frac{\V_{n1}(B_1(t))}{k \sqrt{N_{n1}}} - (p - t) \Bigr)
		+ \rho_n^{\rm M}(t,p) \\
\intertext{and}
	L_n'(t,p) \
	&= \ \sum_{r=1}^k \sqrt{N_{nr}} w_r(p) \V_{nr}(B_r(t))
		- \sum_{r=1}^k N_{nr} w_r(p) (B_r(p) - B_r(t)) \\
	&= \ \sqrt{N_{n1}} w_1(p) \V_{n1}(B_1(t))
		- N_{n1} w_1(p) k (p - t)
		+ \rho_n^{\rm L}(t,p) \\
	&= \ N_{n1} k w_1(p)
			\Bigl( \frac{\V_{n1}(B_1(t))}{k \sqrt{N_{n1}}} - (p - t) \Bigr)
		+ \rho_n^{\rm L}(t,p) ,
\end{align*}
where
\begin{align*}
	|\rho_n^{\rm M}(t,p)| \
	&\le \ o_p(\sqrt{n}) t^{2\delta} + O(n) \max(t,p) (p - t) , \\
	|\rho_n^{\rm L}(t,p)| \
	&\le \ o_p(\sqrt{n}) p^{-1} t^{2\delta} + O(n) p^{-1} \max(t,p) (p - t) .
\end{align*}

Now we proceed similarly as in the proof of Theorem~\ref{thm:Asymptotics1}, defining
\[
	p_n(t) \ := \ t + \frac{\V_n^{(\ell)}(t) + b t^\kappa}{\sqrt{n}}
\]
for some fixed $b \ne 0$. Note that for $t \in (0,c_n]$,
\[
	|p_n(t) - t|
	\ \le \ o_p(n^{-1/2}) t^\delta + O(n^{-1/2}) t^\kappa
	\ = \ o_p(n^{-1/2}) t^\delta ,
\]
because $\kappa > \delta$. Note also that
\[
	t + \frac{\V_n^{(\ell)}(t)}{\sqrt{n}}
	\ = \ t + \frac{\V_{n1}(B_1(t))}{k \sqrt{N_{n1}}}
	\ = \ t - \frac{1 - (1-t)^k}{k} + \frac{\hat{B}_{n1}(t)}{k}
	\ > \ 0 \quad \text{on} \ (0,1) ,
\]
because $\hat{B}_{n1} \ge 0$ and $t \mapsto t - (1 - (1-t)^k)/k$ is strictly convex on $[0,1]$ with derivative $0$ at $0$. Thus $p_n(t) > 0$ for all $t \in (0,c_n]$ in case of $b > 0$.

In case of $b > 0$ we may conclude that
\begin{align*}
	n \Bhat(t) \ - & \sum_{r=1}^k N_{nr} B_r(p_n(t)) \\
	&= \ N_{n1} k \frac{-b t^\kappa}{\sqrt{n}} + \rho_n^{\rm M}(t,p_n(t)) \\
	&\le \ \frac{N_{n1} k}{\sqrt{n}}
		\bigl( -b t^\kappa + o_p(1) t^{2\delta}
			+ O(1) (t + o_p(n^{-1/2}) t^\delta) t^\delta \bigr) \\
	&\le \ \frac{N_{n1} k t^\kappa}{\sqrt{n}}
		\bigl( -b + o_p(1) \bigr) ,
\intertext{and}
	L_n'(t,p_n(t)) \
	&\le \ \frac{N_{n1} k w_1(p) t^\kappa}{\sqrt{n}}
		\bigl( -b + o_p(1) \bigr) .
\end{align*}
Hence for any fixed $b > 0$,
\[
	\Pr \bigl( \sqrt{n} (\BhatZ(t) - t) \le \V_n^{(\ell)}(t) + b t^\kappa
		\ \text{for} \ t \in (0,c_n] \bigr)
	\ \to \ 0 .
\]

Similarly we can show that for any fixed $b < 0$, with asymptotic probability one, $\sqrt{n} (\BhatZ(t) - t) \le \V_n^{(\ell)}(t) + b t^\kappa$ for all $t \in (0,c_n]$.
\end{proof}

\begin{proof}[\bf Proof of Corollary~\ref{cor:Asymptotics}]
It follows from Proposition~\ref{prop:WeightedUniformEP} that
\[
	\sup_{1 \le r \le k, \, u \in [0,1]} |\V_{nr}(u)|
	\ = \ O_p(1) .
\]
Together with \eqref{eq:Asymptotics gamma} this entails that $\sup_{t \in [0,1]} \bigl| \V_n^{\rm Z}(t) - \tilde{\V}_n^{\rm Z}(t) \bigr| \to_p 0$, where $\tilde{\V}_n^{\rm Z} := \sum_{r=1}^k \gamma_r^{\rm Z} \, \V_{nr} \circ B_r$. But $\gamma_r^{\rm Z} \equiv 0$ whenever $\pi_r = 0$. In case of $\pi_r > 0$ it follows from Proposition~\ref{prop:WeightedUniformEP} that $\V_{nr}$ converges in distribution to $\V_r$. Consequently $\tilde{\V}_n^{\rm Z}$ converges in distribution to the Gaussian process $\V^{\rm Z} = \sum_{r=1}^k \gamma_r^{\rm Z} \, \V_r \circ B_r$.
\end{proof}

\begin{proof}[\bf Proof of Theorem~\ref{thm:Asymptotics3}]
The asserted inequalities follow from Jensen's inequality. On the one hand, it follows from $w_r = \beta_r / (B_r(1 - B_r))$ and $\sum_{r=1}^k \beta_r \equiv k$ that
\begin{align*}
	K^{\rm S}(t) \
	&= \ \frac{1}{k} \sum_{r=1}^k \frac{\beta_r(t)}{k} \cdot (\pi_r w_r(t))^{-1} \\
	&\ge \ \frac{1}{k}
		\Bigl( \sum_{r=1}^k \frac{\beta_r(t)}{k} \cdot \pi_r w_r(t) \Bigr)^{-1} \\
	&= \ \Bigl( \sum_{r=1}^k \pi_r \beta_r(t) w_r(t) \Bigr)^{-1}
		\ = \ K^{\rm L}(t) .
\end{align*}
Equality holds if, and only if,
\[
	\pi_1 w_1(t) = \pi_2 w_2(t) = \cdots = \pi_k w_k(t) .
\]
But
\[
	w_1(t) \ = \ \frac{k}{(1-t)( 1 - (1 - t)^k)}
	\quad\text{and}\quad
	w_k(t) \ = \ \frac{k}{t(1 - t^k)} ,
\]
so
\[
	\frac{w_k(t)}{w_1(t)}
	\ = \ \frac{(1-t)(1 - (1-t)^k)}{t(1-t^k)}
	\ = \ \frac{\sum_{j=0}^{k-1} (1-t)^j}{\sum_{j=0}^{k-1} t^j}
\]
is strictly decreasing in $t$. Hence there is at most one solution of the equation $\pi_1 w_1(t) = \pi_k w_k(t)$.

Similarly, with $a_r(t) := \pi_r \beta_r(t) \big/ \sum_{s=1}^k \pi_s \beta_s(t)$,
\begin{align*}
	K^{\rm M}(t) \
	&= \ \sum_{r=1}^k \pi_r \beta_r(t) \cdot w_r(t)^{-1}
		\Big/ \Bigl( \sum_{s=1}^k \pi_s \beta_s(t) \Bigr)^2 \\
	&= \ \sum_{r=1}^k a_r(t) \cdot w_r(t)^{-1}
		\Big/ \sum_{s=1}^k \pi_s \beta_s(t) \\
	&\ge \ \Bigl( \sum_{r=1}^k a_r(t) w_r(t) \Bigr)^{-1}
		\Big/ \sum_{s=1}^k \pi_s \beta_s(t) \\
	&= \ \Bigl( \sum_{r=1}^k \pi_r \beta_r(t) w_r(t) \Bigr)^{-1}
		\ = \ K^{\rm L}(t) .
\end{align*}
Here the inequality is strict unless
\[
	w_1(t) = w_2(t) = \cdots = w_k(t) .
\]
But $w_1(t) = w_k(t)$ implies that $t = 1/2$. Moreover, $w_1(1/2) = 2k/(1 - 2^{-k})$ and
\[
	w_{k-1}(1/2) \ = \ \frac{2k(k-1)}{(k+1)(1 - (k+1)2^{-k})}
\]
are identical if, and only if, $k^2 + k + 2 = 2^{k+1}$. But $2^{k+1} = 2 \sum_{j=0}^k \binom{k}{j}$ is strictly larger than $2(1 + k + k(k-1)/2) = k^2 + k + 2$ if $k \ge 3$.

As to the ratios $E^{\rm Z}(t) := K^{\rm Z}(t)/K^{\rm L}(t)$, note first that
\begin{align*}
	E^{\rm S}(t) \
	&= \ \sum_{r=1}^k \frac{B_r(t)(1 - B_r(t))}{k^2 \pi_r}
		\sum_{s=1}^k \pi_s \beta_s(t) w_s(t) \\
	&\ge \ \min_{r,s=1,\ldots,k} \frac{B_r(t)(1 - B_r(t)) \beta_s(t) w_s(t)}{k^2}
		\Big/ \min_{r=1,\ldots,k} \pi_r \\
	&\to \ \infty \quad\text{as} \
		\min_{r=1,\ldots,k} \pi_r \downarrow 0 .
\end{align*}
On the other hand, with $a_r(t)$ as above,
\[
	E^{\rm M}(t)
	\ = \ \sum_{r=1}^k a_r(t) w_r(t)^{-1} \sum_{s=1}^k a_s(t) w_s(t)
	\ = \ \Ex(W) \Ex(W^{-1})
\]
with a random variable $W$ with distribution $\sum_{r=1}^k a_r(t) \delta_{w_r(t)}$. But with $\ell(t) := \min_r w_r(t)$ and $u(t) := \max_r w_r(t)$, convexity of $w \mapsto w^{-1}$ on $[\ell(t),u(t)]$ implies that
\[
	W^{-1} \ \le \ \frac{W - \ell(t)}{u(t) - \ell(t)} u(t)^{-1}
		+ \frac{u(t) - W}{u(t) - \ell(t)} \ell(t)^{-1} ,
\]
so
\begin{align*}
	\Ex(W) \Ex(W^{-1}) \
	&\le \ \Ex(W) \Bigl( \frac{\Ex(W) - \ell(t)}{u(t) - \ell(t)} u(t)^{-1}
		+ \frac{u(t) - \Ex(W)}{u(t) - \ell(t)} \ell(t)^{-1} \Bigr) \\
	&= \ \frac{\Ex(W) (\ell(t) + u(t) - \Ex(W))}{\ell(t) u(t)} \\
	&\le \ \frac{(\ell(t) + u(t))^2}{4 \ell(t)u(t)}
		\ = \ \frac{\rho(t) + \rho(t)^{-1} + 2}{4} .
\end{align*}
This upper bound for $E^{\rm M}(t)$ is attained approximately, if the distribution of $W$ aproaches the uniform distribution on $\{\ell(t),u(t)\}$. Hence we should choose $(\pi_r)_{r=1}^k$ as follows: Let $r(1), r(2)$ be two different numbers in $\{1,\ldots,k\}$ such that $w_{r(1)}(t) = \ell(t)$ and $w_{r(2)}(t) = u(t)$. Then let
\[
	\pi_r \ \approx \ \begin{cases}
		\beta_{r}(t)^{-1}/(\beta_{r(1)}^{-1} + \beta_{r(2)}^{-1})
		&\text{for} \ r \in \{r(1),r(2)\} , \\
		0
		&\text{for} \ r \not\in \{r(1),r(2)\} .
	\end{cases}
\]
The inequality $\rho(t) \le k$ follows from Lemma~\ref{lem:calculus} and the fact that $\rho(t)$ remains unchanged if we replace $w_r(t)$ with $\tilde{w}_r(t) = t(1-t) w_t(t) \in [1,k]$.
\end{proof}

\bigskip

\paragraph{Acknowledgement.}
Constructive comments by an associate editor and two referees are gratefully acknowledged.


\clearpage

\appendix

\section*{Supplementary material}

\section{Further proofs and technical details}

\begin{proof}[\bf Proof of Lemma~\ref{lem:compute.FhatL}]
Continuity of $L_n(x,\cdot) : [0,1] \to [-\infty,0]$ follows essentially from continuity of $\log : [0,1] \to [-\infty,0]$. For $p \in (0,1)$,
\[
	L_n'(x,p) \ = \ \sum_{r=1}^k N_{nr} \Bigl[ \frac{\beta_r}{B_r}(p) \hat{F}_{nr}(x)
		- \frac{\beta_r}{1 - B_r}(p) (1 - \hat{F}_{nr}(x)) \Bigr] .
\]
It follows from the formula $B_r(p) = \sum_{i=r}^k \binom{k}{i} p^i (1 - p)^{k-i}$ that
\begin{align*}
	\frac{\beta_r}{B_r}(p) \
	&= \ C_r \Big/ \sum_{i=r}^k \binom{k}{i} p^{i+1-r} (1 - p)^{r-i}
\intertext{and}
	\frac{\beta_r}{1 - B_r}(p) \
	&= \ C_r \Big/ \sum_{i=0}^{r-1} \binom{k}{i} p^{i+1-r} (1 - p)^{r-i}
\end{align*}
are strictly decreasing and strictly increasing in $p \in (0,1)$, respectively. Consequently, the derivative $L_n'(x,\cdot)$ is continuous and strictly decreasing on $(0,1)$.

Elementary algebra yields the alternative formula
\[
	L_n'(x,p) \ = \ \sum_{r=1}^k N_{nr} w_r(p) \bigl[ \hat{F}_{nr}(x) - B_r(p) \bigr]
\]
with the auxiliary function
\[
	w_r(p) \ = \ \frac{\beta_r}{B_r(1 - B_r)}(p)
	\ = \ \frac{\beta_r(p)}{B_r(p) B_{k+1-r}(1-p)} .
\]
The latter equation follows from the relation $1 - B_r(p) = B_{k+1-r}(1-p)$ and is highly recommended to avoid rounding errors in case of $p$ being close to $1$. Note also that
\begin{equation}
\label{eq:expansions w}
	w_r(p) \ = \ \begin{cases}
		\displaystyle
		\frac{r + o(1)}{p} & \text{as} \ p \to 0 , \\[2ex]
		\displaystyle
		\frac{k+1-r + o(1)}{1-p} & \text{as} \ p \to 1 .
	\end{cases}
\end{equation}
This implies that the limits of $L_n'(x,\cdot)$ at the boundary of $(0,1)$ satisfy
\begin{align*}
	L_n'(x,0) \ &= \ +\infty \quad\text{if} \ x \ge X_{(1)} , \\
	L_n'(x,1) \ &= \ -\infty \quad\text{if} \ x < X_{(n)} ,
\end{align*}
because $x \ge X_{(1)}$ implies that $\hat{F}_{nr}(x) > 0 = B_r(0)$ for at least one $r$, while $x < X_{(n)}$ implies that $\hat{F}_{nr} < 1 = B_r(1)$ for at least one $r$.
\end{proof}

\begin{proof}[\bf Proof of Lemma~\ref{lem:calculus}]
As to part~(a), note that $w_r$ is a rational and strictly positive function on $(0,1)$. Hence $\tilde{w}_r(t) := t(1-t) w_r(t)$ defines a function with these properties, too. Moreover, it follows from \eqref{eq:expansions w} that $\lim_{t \downarrow 0} \tilde{w}_r(t) = r$ and $\lim_{t \uparrow 1} \tilde{w}_r(t) = k-r+1$. Hence $\tilde{w}_r$ may be viewed as a rational and strictly positive function on a neighborhood of $[0,1]$. In particular, $\tilde{w}_k$ is continuously differentiable on $[0,1]$.

It remains to show that $1 \le \tilde{w}_r \le \max(r,k+1-r)$ on $[0,1]$. The upper bound follows from the fact that for $0 < t < 1$,
\begin{align*}
	\tilde{w}_r(t) \
	&= \ \frac{t(1-t) \beta_r(t)}{B_r(t)}
		+ \frac{t(1-t) \beta_r(t)}{B_{k-r+1}(1-t)} \\
	&= \ \frac{t^r(1-t)^{k-r+1}}{\int_0^t u^{r-1} (1-u)^{k-r} \, du}
		+ \frac{t^r(1-t)^{k-r+1}}{\int_0^{1-t} u^{k-r} (1 - u)^{r-1} \, du} \\
	&\le \ \frac{t^r(1-t)^{k-r+1}}{\int_0^t u^{r-1} \, du \ (1 - t)^{k-r}}
		+ \frac{t^r(1-t)^{k-r+1}}{\int_0^{1-t} u^{k-r} \, du \ t^{r-1}} \\
	&= \ (1 - t) \, r + t \, (k-r+1) \\
	&\le \ \max(r, k-r+1) .
\end{align*}
The lower bound is equivalent to the claim that $\beta_r(t) \ge B_r(t)(1 - B_r(t))/(t(1-t))$ for any $t \in (0,1)$. Since $\log \beta_r(u) = \log C_r + (r-1) \log u + (k-r) \log (1 - u)$ is concave in $u \in (0,1)$, this assertion follows from Lemma~\ref{lem:logconcave.on.01} below.

For proving part~(b), note first that $|p - t| \le c t(1-t)$ implies the inequalities $p \le (1+c)t$ and $1-p \le (1+c)(1-t)$. Moreover, since $\bigl| p(1-p) - t(1-t) \bigr| \le |p-t|$, we may conclude that $p(1-p) \ge (1-c) t(1-t)$. Consequently,
\begin{align*}
	\Bigl| \frac{w_r(p)}{w_r(t)} - 1 \Bigr| \
	&= \ \frac{\bigl| \tilde{w}_r(p) t(1-t) - \tilde{w}_r(t) p(1-p) \bigr|}
		{\tilde{w}_r(t) p(1-p)} \\
	&\le \ \frac{\bigl| \tilde{w}_r(p) - \tilde{w}_r(t) \bigr| t(1-t)
			+ \tilde{w}_r(t) \bigl| t(1-t) - p(1-p) \bigr|}
		{\tilde{w}_r(t) p(1-p)} \\
	&\le \ \frac{\bigl| \tilde{w}_r(p) - \tilde{w}_r(t) \bigr|/4 + C_w |t - p|}
		{c_w (1 - c) t(1-t)} \\
	&\le \ \frac{c_w'/4 + C_w}{c_w (1-c)} \, \frac{|p - t|}{t(1-t)} ,
\end{align*}
where $c_w' := \max_{1 \le r \le k, u \in [0,1]} |\tilde{w}_r'(u)|$. Moreover, for $\min(t,p) \le \xi \le \max(t,p)$,
\[
	\frac{|\beta_r'(\xi)|}{\beta_r(\xi)}
	\ = \ \frac{|r-1 - (k-1) \xi|}{\xi(1-\xi)}
	\ \le \ \frac{k-1}{(1-c)t(1-t)}
	\quad\text{and}\quad
	\frac{\beta_r(\xi)}{\beta_r(t)} \ \le \ (1 + c)^{k-1} .
\]
Hence Taylor's formula shows that for a suitable such $\xi$,
\[
	\Bigl| \frac{B_r(p) - B_r(t)}{\beta_r(t)(p-t)} - 1 \Bigr|
	\ = \ \frac{|\beta_r'(\xi)| |p-t|}{2 \beta_r(t)}
	\ \le \ \frac{(k-1)(1+c)^{k-1}}{c-1} \, \frac{|p-t|}{t(1-t)} .
\]\\[-7ex]
\end{proof}

In the proof of Lemma~\ref{lem:calculus} we referred to the following general inequality which is possibly of independent interest:

\begin{Lemma}
\label{lem:logconcave.on.01}
Let $\beta$ be a strictly positive probability density on $(0,1)$ such that $\log \beta$ is concave. Then its distribution function $B : [0,1] \to [0,1]$ satisfies the following inequalities: For any $t \in (0,1)$,
\[
	\beta(t) \ \ge \ \frac{B(t)(1 - B(t))}{t(1-t)}
\]
with equality if, and only if, $\beta \equiv 1$.
\end{Lemma}

\begin{proof}[\bf Proof of Lemma~\ref{lem:logconcave.on.01}]
For $a \in \R$ let $G_a : [0,1] \to [0,1]$ be the distribution function given by
\[
	G_a(x) \ := \ \begin{cases}
		(e^{ax} - 1)/(e^a - 1) & \text{if} \ a \ne 0 , \\
		x & \text{if} \ a = 0 .
	\end{cases}
\]
Then $G_a$ has log-linear density
\[
	g_a(x) := G_a'(x) \ = \ e^{ax - c(a)}
\]
with $c(0) = 0$ and $c(a) = \log((e^a - 1)/a)$ for $a \ne 0$. For fixed $t \in (0,1)$, $G_a(t)$ is continuous in $a \in \R$ with $\lim_{a \ge \infty} G_a(t) = 0$ and $\lim_{a \to -\infty} G_a(t) = 1$. Hence for a suitable $a = a(t) \in \R$,
\[
	B(t) \ = \ G_a(t) .
\]
If we fix this value $a$, then the previous equality implies that $\beta(s) \ge g_a(s)$ for some $s \in (0,t)$ and $\beta(u) \ge g_a(u)$ for some $u \in (t,1)$. But then concavity of $\log \beta$ and linearity of $\log g_a$ yield the inequality $\beta(t) \ge g_a(t)$. Moreover, if $\beta(t) = g_a(t)$, then $\beta \le g_a$, and this implies that $\beta \equiv g_a$. Hence it suffices to prove the claim in case of $\beta \equiv g_a$ for some $a \in \R$.

Since $g_0 \equiv 1$ and $G_0(t) = t$, the asserted inequality is an equality in case of $a = 0$. Hence it remains to show that $G_a(t)(1 - G_a(t)) < t(1-t) g_a(t)$ in case of $a \ne 0$. Indeed,
\begin{align*}
	\frac{G_a(t)(1 - G_a(t))}{t (1 - t) g_a(t)} \
	&= \ \frac{(e^{at} - 1)(e^a - e^{at})}
		{t(1 - t) e^{at} a (e^a - 1)} \\
	&= \ \frac{e^{at} - 1}{at} \cdot \frac{e^{a(1 - t)} - 1}{a(1 - t)}
		\Big/ \frac{e^a - 1}{a}
	&= \ \exp \bigl( h(at) + h(a - at) - h(a) \bigr) ,
\end{align*}
where $h(x) := \log((e^x - 1)/x)$ for $x \ne 0$. In case of $a > 0$ it follows from $\lim_{x \to 0} h(x) = 0$ that
\[
	h(at) + h(a(1 - t)) - h(a)
	\ = \ \int_0^{at} \bigl( h'(u) - h'(a(1 - t) + u) \bigr) \, du
	\ < \ 0 ,
\]
because $h''(x) = x^{-2} - (e^x + e^{-x} - 2)^{-1} > 0$, so $h'$ is strictly increasing. In case of $a < 0$, it follows from $h(x) = x + h(-x)$ that
\[
	h(at) + h(a(1 - t)) - h(a)
	\ = \ h(|a|t) + h(|a|(1 - t)) - h(|a|)
	\ < \ 0
\]
as well.
\end{proof}

\paragraph{Details about asymptotic variances and the function $\rho$ in case of $k = 2$.}
In the special case $k = 2$, elementary calculations reveal that
\begin{align*}
	&\beta_1(t) \ = \ 1 - u , \quad
		B_1(1 - B_1)(t) \ = \ K(t) \frac{3 - 4u + u^2}{4} , \quad
		w_1(t) \ = \ \frac{4}{K(t) (3 - u)} , \\
	&\beta_2(t) \ = \ 1 + u , \quad
		B_2(1 - B_2)(t) \ = \ K(t) \frac{3 + 4u + u^2}{4} , \quad
		w_2(t) \ = \ \frac{2}{K(t)(3+u)} ,
\end{align*}
where $u := 2t - 1 \in [-1,1]$ and $K(t) := K(t,t) = t(1 - t)$. In particular,
\[
	\tilde{w}_1(t) \ = \ \frac{4}{3 - u}, \quad
	\tilde{w}_2(t) \ = \ \frac{4}{3 + u}
	\quad\text{and}\quad
	\frac{\rho(t) + \rho(t)^{-1} + 2}{4}
	\ = \ \frac{9}{9 - u^2} .
\]
Moreover, with $\Delta := \pi_2 - \pi_1$ these formulae entail that
\begin{align*}
	K^{\rm S}(t) \
	&= \ \frac{K(t)}{4} \frac{3 + u^2 - 4u\Delta}{(1 - \Delta^2)} , \\
	K^{\rm M}(t) \
	&= \ \frac{K(t)}{4}
		\frac{3 + u^2 + 4u\Delta}{(1 + u\Delta)^2} , \\
	K^{\rm L}(t) \
	&= \ \frac{K(t)}{4}
		\frac{9 - u^2}
		     {3 - u^2 + 2u\Delta} .
\end{align*}
The top left panel in Figure~\ref{fig:Asymptotics2} shows for $\pi_1 = \pi_2 = 1/2$ the asymptotic variances $K^{\rm M}(t) = K^{\rm S}(t) > K^{\rm L}(t)$ as well as the variances $K(t)$ for simple random sampling. In the top right and lower panels one sees for $\pi_1 = 1 - \pi_2 = 1/2, 5/8, 3/4$ the relative asymptotic efficiencies $E^{\rm M}(t) = K^{\rm M}(t)/K^{\rm L}(t)$ and $E^{\rm S}(t) = K^{\rm S}(t)/K^{\rm L}(t)$ of $\BhatL$ with respect to $\BhatM$ and $\BhatS$, respectively. In each panel the gray dotted line depicts the upper bound $E_{\rm max}^M(t) = (\rho(t) + \rho(t)^{-1})/4 \le 1.125$ for $E^{\rm M}(t)$. Note that $E^{\rm S}(t)$ can get arbitrarily large.

\begin{figure}
\includegraphics[width=0.49\textwidth]{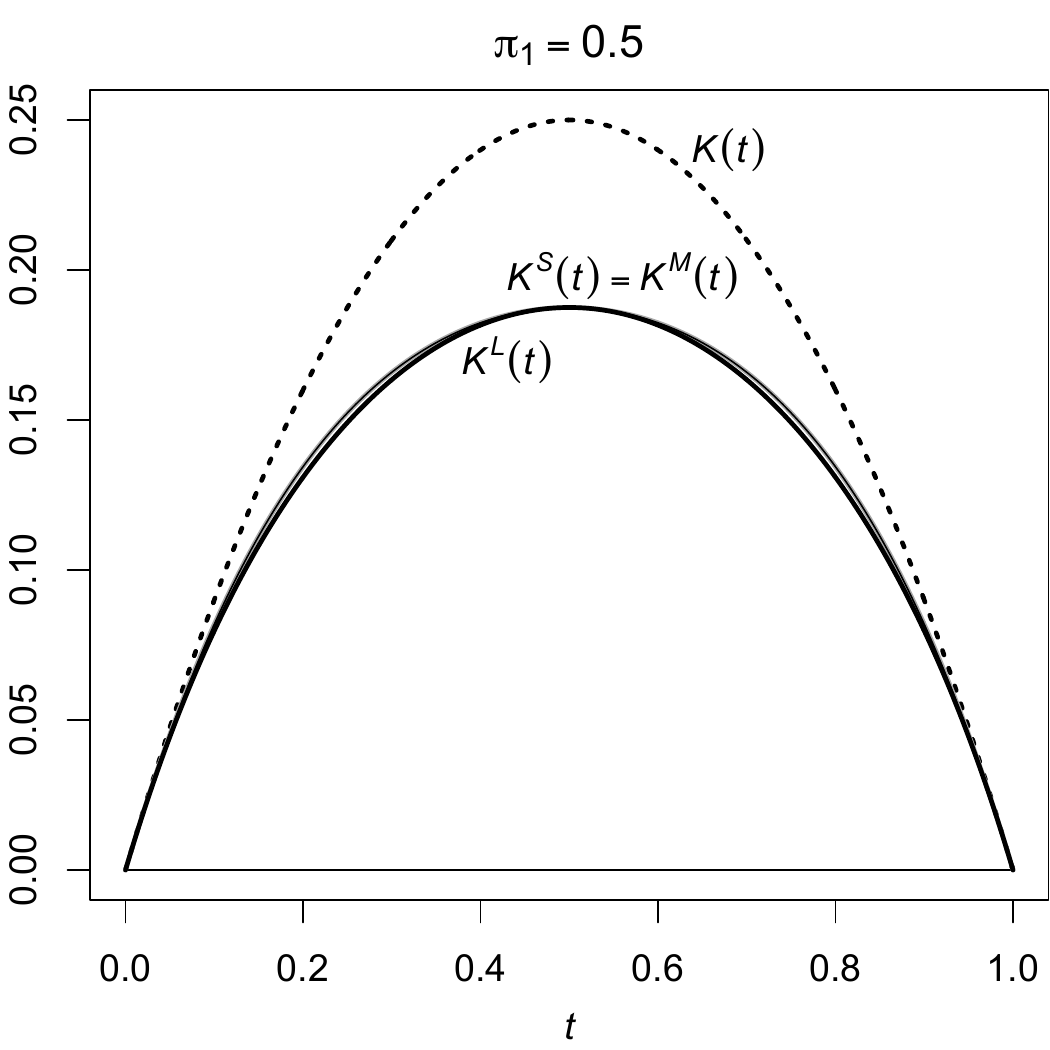}
\hfill
\includegraphics[width=0.49\textwidth]{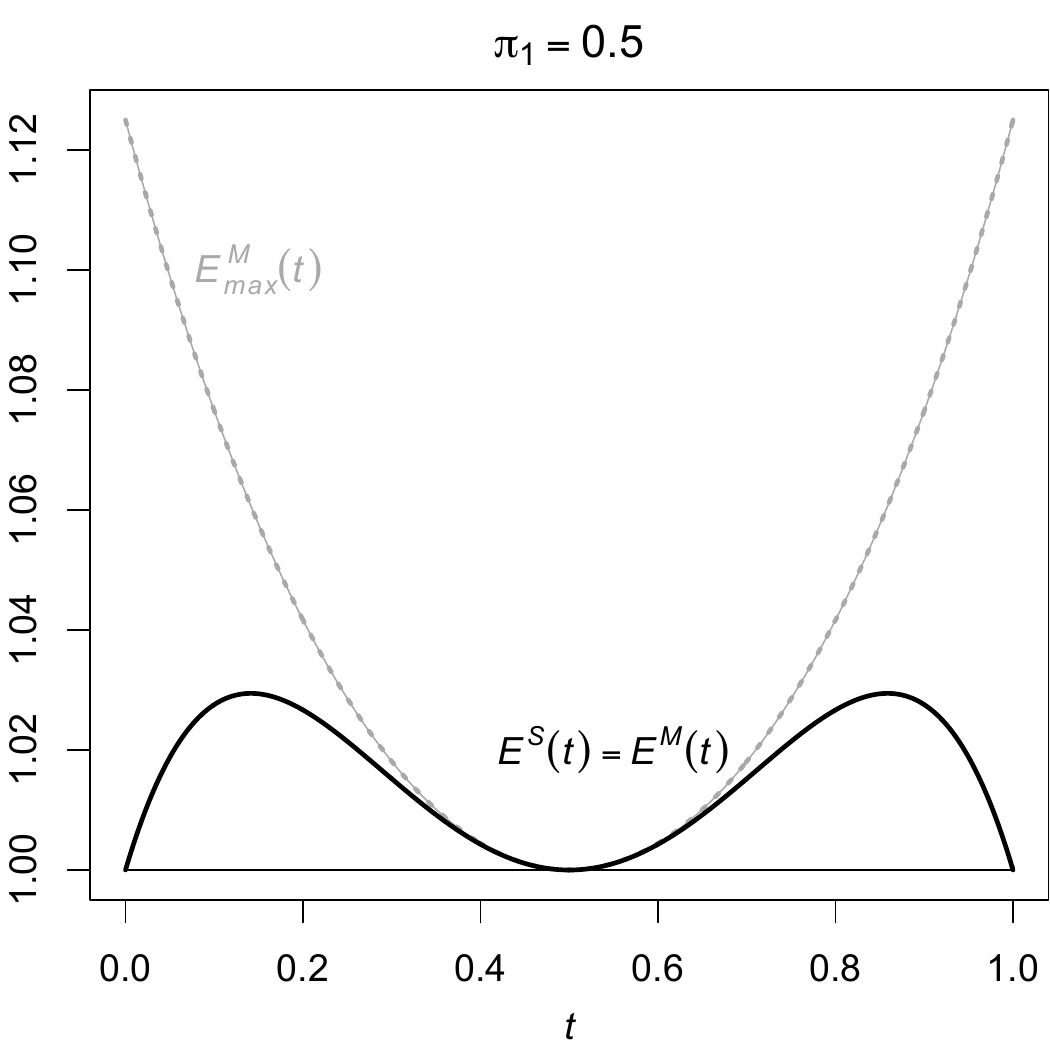}

\includegraphics[width=0.49\textwidth]{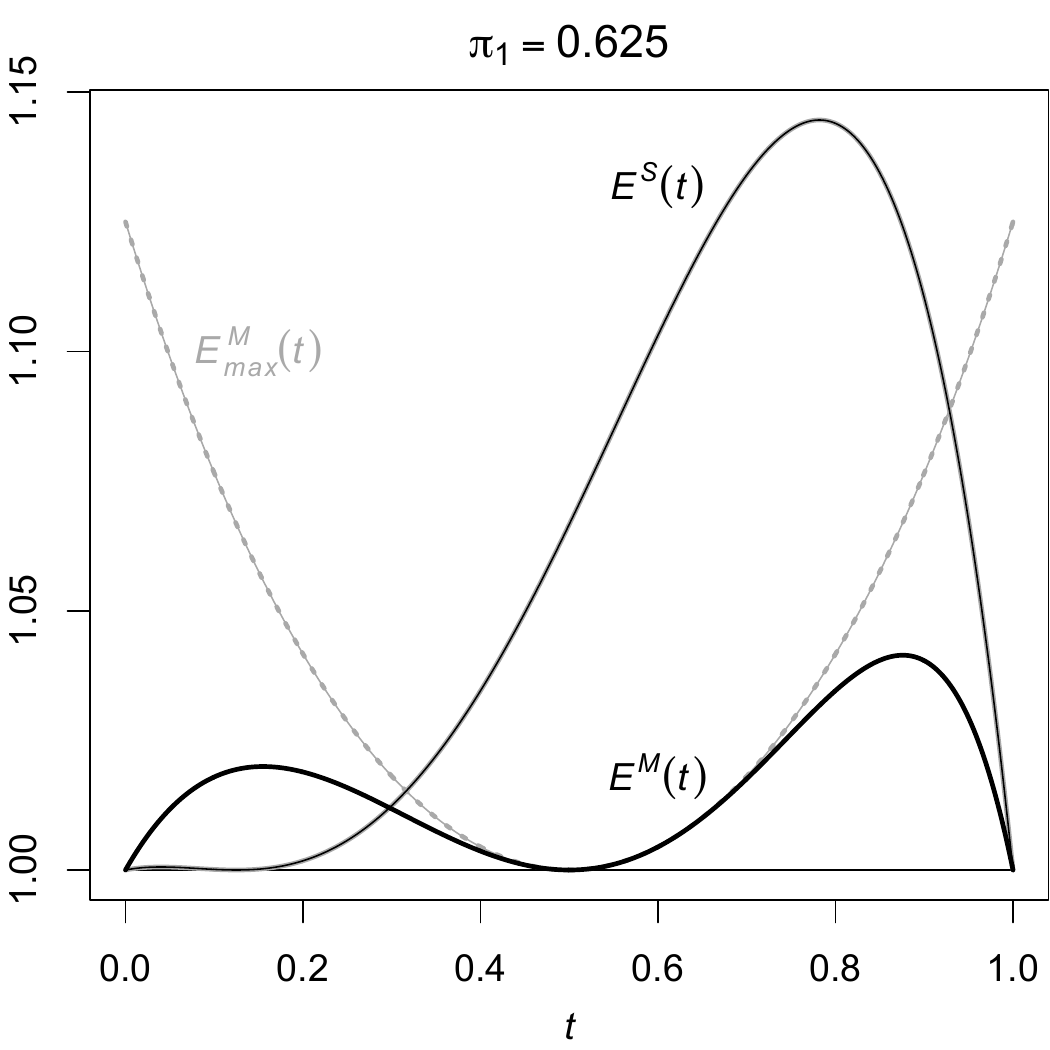}
\hfill
\includegraphics[width=0.49\textwidth]{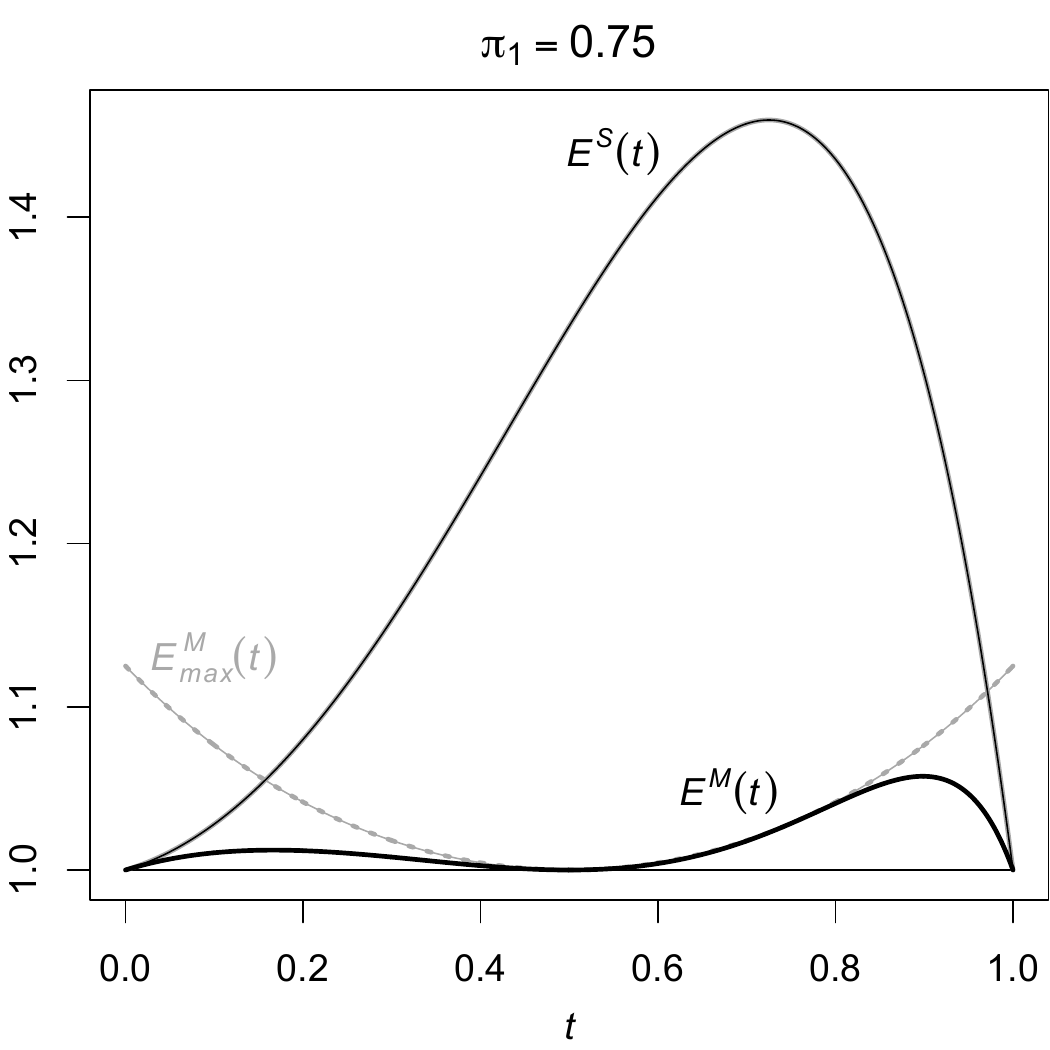}

\caption{Asymptotic variances and relative efficiencies for $k = 2$.}
\label{fig:Asymptotics2}
\end{figure}

\begin{proof}[\bf Proof of \eqref{eq:Asymptotics.pw1} and \eqref{eq:Asymptotics.pw2}.]
Let $(\hat{p}_n)_n, (\hat{q}_n)_n$ be random sequences in $[0,1]$ converging to $p := F(x)$ in probability. It follows from Lindeberg's Central Limit theorem, applied to convolutions of binomial distributions, that
\[
	\sup_{y \in \R} \, \Bigl| G_{\Nn,\hat{p}_n}(y)
		- \Phi \Bigl( \frac{\sqrt{n}}{\sigma(p)}
				\Bigl( \frac{y}{n} - \mu_n(\hat{p}_n) \Bigr) \Bigr)
		\Bigr|
	\ \to_p \ 0 ,  
\]
where
\begin{align*}
	\mu_n(q) \
	&:= \ \sum_{r=1}^k \frac{N_{nr}}{n} B_r(q)
		\quad\text{for} \ q \in [0,1] , \\
	\sigma(p) \
	&:= \ \Bigl( \sum_{r=1}^k \pi_r B_r(p)(1 - B_r(p)) \Bigr)^{1/2} .
\end{align*}
Moreover,
\begin{align*}
	\mu_n(\hat{q}_n) - \mu_n(\hat{p}_n) \
	&= \ \sum_{r=1}^k \frac{N_{nr}}{n} \bigl( B_r(\hat{q}_n) - B_r(\hat{p}_n) \bigr) \\
	&= \ \Bigl( \sum_{r=1}^k \pi_r \beta_r(p) + o_p(1) \Bigr) (\hat{q}_n - \hat{p}_n) \\
	&= \ \Bigl( \frac{\sigma(p)}{K^{\rm M}(F(x))^{1/2}} + o_p(1) \Bigr)
		(\hat{q}_n - \hat{p}_n) .
\end{align*}
Now we apply these findings to
\[
	\hat{p}_n \ := \ \FhatM(x) + \frac{\Delta}{\sqrt{n}}
	\quad\text{and}\quad
	\hat{q}_n \ := \ \FhatM(x)
\]
with $\Delta \in \R$ to be specified later. Note that $\mu_n(\hat{q}_n) = \Fhat(x)$ by definition of $\FhatM(x)$. Hence for $c = 0, 1$,
\begin{align*}
	G_{\Nn,\hat{p}_n}(n \Fhat(x) - c) \
	&= \ \Phi \Bigl( \frac{\sqrt{n}}{\sigma(p)}
		\bigl( \Fhat(x) + O(n^{-1}) - \mu_n(\hat{p}_n) \bigr) \Bigr) + o_p(1) \\
	&= \ \Phi \Bigl( \frac{\sqrt{n}}{\sigma(p)}
		\bigl( \mu_n(\hat{q}_n) - \mu_n(\hat{p}_n) \bigr) \Bigr) + o_p(1) \\
	&\to_p \ \Phi \Bigl( \frac{- \Delta}{K^{\rm M}(F(x))^{1/2}} \Bigr) .
\end{align*}
If we choose $\Delta$ strictly smaller or strictly larger than $K^{\rm M}(F(x))^{1/2} \Phi^{-1}(1 - \alpha)$, then the limit of $G_{\Nn,\hat{p}_n}(n \Fhat(x))$ is strictly larger or strictly smaller than $\alpha$, respectively. This proves \eqref{eq:Asymptotics.pw2}. If we choose $\Delta$ strictly smaller or strictly larger than $- K^{\rm M}(F(x))^{1/2} \Phi^{-1}(1 - \alpha)$, then the limit of $G_{\Nn,\hat{p}_n}(n\Fhat(x) - 1)$ is strictly larger or strictly smaller than $1 - \alpha$, respectively, which proves \eqref{eq:Asymptotics.pw1}.
\end{proof}

\section{Computer code}

On the first author's web page (www.stat.unibe.ch/duembgen) one can download specific computer programs for the methods and examples presented here. All code is for the statistical computing environment \nocite{R_2013}{R}. The files are:

\begin{itemize}
\item Extimation.R: \ Computation of the point estimators $\FhatS$, $\FhatM$ and $\FhatL$.
\item Simulations.R: \ Simulation of RSS and JPS data sets, including sampling from the Dell--Clutter model.
\item ConfBands.R: \ Computing pointwise and simultaneaous confidence bands for $F$.
\item MonteCarlo.R: \ Monte Carlo estimation of the estimators' bias and RMSE; simulating sampling from a finite population as in Section~\ref{subsec:Municip_CH}.
\item Municip\_CH\_2015.txt: \ Data for Section~\ref{subsec:Municip_CH}.
\item MainScript.R: \ Main script file with examples for all procedures coded in the previous R files.
\end{itemize}
\end{document}